%% file: arxiv.tex
\documentclass[11pt]{article}
\usepackage[
  letterpaper,
  top=1in,
  bottom=1in,
  left=1in,
  right=1in
]{geometry}
\usepackage{mathpazo}
\usepackage{amsmath,amssymb,amsthm,mathtools}
\usepackage{aliascnt}
\usepackage{microtype}
\usepackage{enumitem}
\usepackage[dvipsnames]{xcolor}
\usepackage{booktabs}
\usepackage{graphicx}
\usepackage{adjustbox}
\usepackage{float}
\usepackage{tikz}
\usetikzlibrary{arrows.meta,backgrounds,calc,fit}
\usepackage{fancyhdr}
\usepackage{appendix}
\usepackage[pagebackref]{hyperref}
\usepackage[nameinlink,capitalize]{cleveref}

\hypersetup{
  colorlinks=true,
  urlcolor=blue,
  linkcolor=RoyalBlue,
  citecolor=OliveGreen,
  linktocpage=true
}

\allowdisplaybreaks

\newtheorem{theorem}{Theorem}[section]

\newaliascnt{proposition}{theorem}
\newtheorem{proposition}[proposition]{Proposition}
\aliascntresetthe{proposition}

\newaliascnt{lemma}{theorem}
\newtheorem{lemma}[lemma]{Lemma}
\aliascntresetthe{lemma}

\newaliascnt{corollary}{theorem}

\aliascntresetthe{corollary}

\newtheoremstyle{conjecturestyle}
  {\topsep}
  {0pt}
  {\itshape}
  {}
  {\bfseries}
  {.}
  {0.5em}
  {}
\theoremstyle{conjecturestyle}
\newtheorem*{conjecture}{Conjecture}

\theoremstyle{definition}
\newaliascnt{definition}{theorem}
\newtheorem{definition}[definition]{Definition}
\aliascntresetthe{definition}

\theoremstyle{remark}
\newaliascnt{remark}{theorem}
\newtheorem{remark}[remark]{Remark}
\aliascntresetthe{remark}

\theoremstyle{plain}
\newaliascnt{claim}{theorem}

\aliascntresetthe{claim}
\numberwithin{equation}{section}

\newcommand{\Aut}{\operatorname{Aut}}

\title{Rigorous Low-Degree Implications for Planted Subgraph Detection: Noise and Treewidth}
\author{Xuan Chen\thanks{Peking University. \texttt{chenxuan@stu.pku.edu.cn}.}
\and
Shuangping Li\thanks{Yale University. \texttt{shuangping.li@yale.edu}.}}
\date{\today}

\begin{document}
\hypersetup{pageanchor=false}
\maketitle

\begin{abstract}
The low-degree method predicts computational thresholds by asking whether
low-degree polynomials can distinguish a planted distribution from a null
one. Recent counterexamples by Buhai, Hsieh, Jain, and
Kothari (FOCS 2025)~\cite{BuhaiHsiehJainKothari2025} and
Mao (2026)~\cite{Mao2026} show that low-degree indistinguishability does not, in
general, rule out efficient algorithms, even on noisy inputs. Hsieh, Kane,
Kothari, Li, Mohanty, and Tiegel (STOC 2026)~\cite{HsiehEtAl2026} asked what it does imply, and
showed that for permutation-invariant distributions on \(\{0,1\}^n\),
low-degree indistinguishability from the uniform distribution implies
statistical indistinguishability after noise. In this
work, we study the same question for planted subgraphs in sparse random
graphs.

Let \(Q_n=G(n,c/n)\), and let \(P_n\) be obtained by planting a uniformly
random copy of a deterministic graph \(\Gamma_n\) into an independent sample
from \(Q_n\). In the supercritical regime \(c>1\), we show that if \(P_n\)
is degree-\(D_n\) indistinguishable from \(Q_n\) and the treewidth of
\(\Gamma_n\) is \(o(D_n/\log n)\), then a noisy version of \(P_n\) is
asymptotically indistinguishable from \(Q_n\). In the critical and
subcritical regimes \(0<c\leq 1\), we obtain the same conclusion without any
treewidth assumption, provided \(D_n=\omega(\log n)\).

Our proof has two main ingredients. First, we identify the subgraph-count
and automorphism factors in the Fourier expansion with counts of isomorphism
triples. Second, we cut a tree decomposition of \(\Gamma_n\) into subtrees,
which splits each large Fourier support into low-degree pieces that meet at
only a few interface vertices, and we use noise to absorb the cost of
reassembling them. At and below criticality, we use the low-degree
assumption to rule out short cycles and noise to destroy the remaining long
ones.
\end{abstract}

\thispagestyle{empty}

\newpage

\tableofcontents

\thispagestyle{empty}

\newpage
\hypersetup{pageanchor=true}
\setcounter{page}{1}

\section{Introduction}
\label{section:intro}
Many average-case inference problems exhibit a gap between what is
statistically possible and what known efficient algorithms achieve. A central
goal of average-case complexity is to explain and predict such
statistical--computational gaps. One approach, modeled on worst-case
complexity, builds reductions from canonical average-case problems. This
program has made substantial progress, but such reductions are delicate
because they must preserve the target distribution
\cite{BerthetRigollet2013,BrennanBresler2020}. A complementary approach proves
lower bounds against restricted classes of algorithms, such as spectral and
low-degree methods, statistical-query algorithms, Markov chains, local and
stable algorithms, and convex hierarchies like Sum-of-Squares
\cite{Jerrum1992,HopkinsEtAl2017,KuniskyWeinBandeira2022}. Each class comes
with its own tools: pseudocalibration links Sum-of-Squares lower bounds to
low-degree moment matching, and overlap-gap properties rule out broad
families of local and input-stable algorithms for random optimization
problems \cite{BarakEtAl2019,Hopkins2018,Gamarnik2021}. For statistical-query
algorithms the link to low degree is rigorous: under mild conditions, they
are essentially equivalent to low-degree tests \cite{BrennanEtAl2021}. Among
these frameworks, the low-degree method is particularly attractive, because
its calculations are often tractable and its predictions match the best
known algorithms for many problems.

The low-degree framework compares a planted distribution $P$ with a null
distribution $Q$ through their likelihood ratio $L=dP/dQ$. Let $L^{\le D}$, the degree-$D$ low-degree likelihood ratio
(LDLR), denote the orthogonal projection of $L$ onto the polynomials of
degree at most $D$ in the observed coordinates. Then
\begin{equation}
    \left\lVert L^{\le D}\right\rVert_{L^2(Q)}^2-1
    =\sup_{\substack{f:\,\deg(f)\le D\\ \operatorname{Var}_Q(f)>0}}
    \frac{\bigl(\mathbb E_P[f]-\mathbb E_Q[f]\bigr)^2}
    {\operatorname{Var}_Q(f)}.
    \label{eq:intro-low-degree-advantage}
\end{equation}
The square root of the left-hand side, often called the
\emph{degree-$D$ advantage}, is the largest difference between the means
of a degree-$D$ polynomial under $P$ and under $Q$, normalized by its
standard deviation under $Q$.
Thus $\|L^{\le D}\|_{L^2(Q)}^2-1=o(1)$ means that no degree-$D$ statistic
achieves a nonvanishing advantage. Because of its connections to spectral
methods, the Sum-of-Squares hierarchy, and pseudocalibration, and because of
its empirical track record, vanishing low-degree advantage is now widely used
as evidence of computational hardness
\cite{Hopkins2018,HopkinsEtAl2017,KuniskyWeinBandeira2022}.

The success of the low-degree method is somewhat surprising. The LDLR sees
only the low-degree part of the likelihood ratio, while the high-degree
Fourier coefficients may carry a great deal of information that an efficient
algorithm could, in principle, exploit. Noise helps reconcile the two. In the
Bernoulli setting, for \(\varepsilon\in(0,1)\), let \(T_\varepsilon P\)
denote the law obtained by independently resampling each coordinate from its
null marginal with probability \(\varepsilon\). Under a product Bernoulli
null, this operator acts diagonally on the Fourier expansion and multiplies
every degree-\(d\) coefficient by \((1-\varepsilon)^d\) \cite{MR3443800}.
Noise therefore damps the high-degree information that the LDLR does not
control, while largely preserving the low-degree structure that it measures.

This intuition is behind the low-degree conjecture, formalized in Hopkins's
thesis \cite[Conjecture~2.2.4]{Hopkins2018}. We state it in the informal form
of \cite[Conjecture~1.1]{HsiehEtAl2026}.

\begin{conjecture}[The Low-Degree Hypothesis (informal)]

Fix \(d\in\mathbb N\). Let \(Q\) be the uniform distribution on
\(\{0,1\}^{\binom{n}{d}}\) viewed as symmetric
\(d\)-tensors. Let $P$ be a distribution on the same domain
as $Q$ that is \(S_n\)-symmetric, meaning invariant under the
relabeling action of \(S_n\), and suppose that
\[
    \chi^{2}_{D}(P\,\|\,Q)
    =\|L^{\le D}\|_{L^2(Q)}^2-1=o(1)
    \qquad\text{for}\qquad D=(\log n)^{1+\Omega(1)}.
\]
Then for every fixed \(\varepsilon\in(0,1)\), no polynomial-time computable test distinguishes \(T_\varepsilon P\) from \(Q\) with probability
\(1-o(1)\).
\end{conjecture}

\noindent A stronger degree-versus-runtime formulation, also due to Hopkins
\cite{Hopkins2018}, predicts that vanishing degree-\(D\) advantage implies
hardness against algorithms running in time \(n^{\widetilde O(D)}\).

The low-degree heuristic has correctly predicted the computational thresholds
of many problems, including planted clique and planted dense subgraph or
submatrix, sparse PCA and spiked matrix models, tensor decomposition,
community detection and graphon estimation, correlated-graph detection, group
testing, sparse linear regression, random \(k\)-SAT (up to a constant factor), and
optimization on
random graphs
\cite{HopkinsEtAl2017,KuniskyWeinBandeira2022,HopkinsSteurer2017,
DingKuniskyWeinBandeira2023,SchrammWein2022,
BandeiraBanksKuniskyMooreWein2021,
BandeiraElAlaouiHopkinsSchrammWeinZadik2022,CojaOghlanEtAl2022,
BreslerHuang2022,Wein2022,Wein2023,DingDuLi2023,LuoGao2024,
YuZadikZhang2025,SohnWein2025,Wein2025Survey,bok2026detection}.
Its limits have also become clearer. Holmgren and Wein
\cite{HolmgrenWein2021} refuted Hopkins's original formulation by exploiting
the specific noise operator, and showed that the symmetry assumption is
necessary. Jia and Vijayaraghavan \cite{JiaVijayaraghavan2026} exploit
anti-concentration in robust subspace recovery, outside the product-null
setting. These examples show that both the choice of noise and the structural
assumptions in the conjecture matter. More recently,
Buhai, Hsieh, Jain, and Kothari \cite{BuhaiHsiehJainKothari2025} refuted the
quasipolynomial-time formulation using list decoding of noisy polynomial
interpolation, and Mao \cite{Mao2026} refuted the polynomial-time formulation
for graphs, with a planted distribution built from a Reed--Muller code that a
deterministic rank test detects after noise.

These counterexamples leave open what vanishing low-degree advantage
rigorously implies. Hsieh, Kane, Kothari,
Li, Mohanty, and Tiegel \cite{HsiehEtAl2026} initiated a systematic study of
this question and developed tools that turn LDLR bounds into rigorous lower
bounds against concrete algorithms. For Boolean vectors whose law is
invariant under all permutations of the coordinates, with a product Bernoulli
null of fixed bias, they show that vanishing LDLR at logarithmic degree
forces the noisy planted distribution to converge to the null in total
variation. After noise, therefore, no test can distinguish
the two distributions with nonvanishing advantage.

They also prove two Gaussian results. For Gaussian vectors, under the
quantitative low-degree assumptions of their theorem, the noisy planted and
null models give statistically indistinguishable joint laws for the
evaluations of all symmetric polynomials of degree up to
$O(\log n/\log\log n)$; the required rate of low-degree indistinguishability
grows with the target degree. For symmetric Gaussian matrices, vanishing
low-degree advantage up to degree at least $\log n\log\log n$ yields the
analogous conclusion, after noise, for each fixed signed count of a connected
subgraph.

None of these results covers sparse random graphs. The Boolean result
requires a null of fixed bias and invariance under all permutations of the coordinates,
while \(G(n,c/n)\) has vanishing edge density and a planted subgraph is
invariant only under relabeling the vertices. The Gaussian-matrix result is
closer to graph data, but it controls only one fixed subgraph count at a time.

In this paper, we study this question for planted subgraphs in sparse random
graphs. We start from an Erd\H{o}s--R\'enyi graph with constant average
degree, add a uniformly relabeled copy of a deterministic graph, and ask
whether the result remains indistinguishable from the null after independent
edge resampling. We show that vanishing low-degree advantage forces the
entire noised graph to be indistinguishable from the null for broad classes
of planted graphs. Above the Erd\H{o}s--R\'enyi critical point this holds
under a treewidth condition, and at and below criticality no treewidth
assumption is needed.

\subsection{The Planted-Subgraph Model and Main Results}

Fix $c>0$, and let the null distribution be the Erd\H{o}s--R\'enyi graph
\[
    Q_n=G(n,c/n).
\]
Given a deterministic simple graph $\Gamma_n$ with at most $n$ vertices,
we obtain the planted distribution $P_n$ by choosing a uniformly random
injective map $\phi_n:V(\Gamma_n)\to[n]$ and adding the edges of
$\phi_n(\Gamma_n)$ to an independent graph sampled from $Q_n$. We write
\[
    L_n=\frac{dP_n}{dQ_n}.
\]
For $\varepsilon\in(0,1)$, the operator $T_\varepsilon$ independently
resamples every edge indicator from its null marginal with probability
$\varepsilon$. Since this operation preserves $Q_n$, we can equivalently
generate $T_\varepsilon P_n$ by deleting each edge of $\Gamma_n$
independently with probability $\varepsilon$ before planting it. Let
\[
    L_{n,\varepsilon}:=
    \frac{d(T_\varepsilon P_n)}{dQ_n}
\]
denote the noised likelihood ratio. We study the testing problem
\[
    H_0:G_n\sim Q_n
    \qquad\text{versus}\qquad
    H_1:G_n\sim T_\varepsilon P_n.
\]

The high-degree terms of $L_n$ encode large planted patterns. To compare them
with the low-degree truncation, we cut these patterns into small pieces
without creating too many shared boundary vertices, and treewidth is the
parameter that makes this possible.

A tree decomposition of a graph $F$ is a tree whose nodes carry vertex sets,
called bags, such that every vertex of $F$ lies in some bag, both endpoints
of every edge lie in a common bag, and the bags containing any fixed vertex
form a connected subtree. Its width is the largest bag size minus one, and
the \emph{treewidth} $\operatorname{tw}(F)$ is the minimum width over all
tree decompositions \cite{RobertsonSeymour1986,Bodlaender1993}. Removing an
edge of the decomposition tree splits the graph into two parts that interact
only through the intersection of two adjacent bags, so a graph of small
treewidth can be cut repeatedly along small vertex interfaces.

Our first result shows that if $\Gamma_n$ has
sufficiently small treewidth and the low-degree advantage vanishes, then
$T_\varepsilon P_n$ is close to $Q_n$ in total variation.

\begin{theorem}[Supercritical Regime]
    Fix $c>0$, and let $\{\Gamma_n\}_{n\ge1}$ be a sequence of simple
    graphs with at most $n$ vertices. Suppose there exists a sequence of positive integers $\{D_n\}_{n\ge1}$ such that
    \begin{equation}
        \operatorname{tw}(\Gamma_n)=o(D_n/\log n)
        \label{eq:tree-width-condition}
    \end{equation}
    and
    \begin{equation}
        \left\lVert L_n^{\le D_n}\right\rVert^2_{L^2(Q_n)}=1+o(1).
        \label{eq:low-degree-condition}
    \end{equation}
    Then, for every fixed $\varepsilon\in(0,1)$,
    \begin{equation}
        \left\lVert L_{n,\varepsilon}\right\rVert_{L^2(Q_n)}^2=1+o(1).
        \label{eq:noised-L2-consequence}
    \end{equation}
    Consequently,
    \begin{equation}
        \left\lVert T_{\varepsilon}P_n-Q_n\right\rVert_{\mathrm{TV}}=o(1).
        \label{eq:TV-consequence}
    \end{equation}
    \label{thm:main}
\end{theorem}
The conclusion is information-theoretic: no test, efficient or not, can
distinguish $T_\varepsilon P_n$ from $Q_n$ with nonvanishing advantage. The
noise cannot be omitted in general, since low-degree indistinguishability
alone does not imply closeness in total variation.

\vspace{-0.1in}
\paragraph{Consequences of \cref{thm:main}.}
If $D_n=\omega(\log n)$, then \cref{thm:main} applies to every
bounded-treewidth planted family that satisfies the low-degree condition \eqref{eq:low-degree-condition}. This
includes all forests, such as paths, stars, caterpillars, and matchings, as
well as high-girth unicyclic, cactus, outerplanar, and series--parallel
graphs \cite{Bodlaender1993,Bodlaender1998}. Subdividing edges does not
change treewidth, so subdividing the edges of cycles, theta graphs, ladders,
and fixed-width grid strips gives natural examples with long cycles but no
short ones.

The condition also allows $\operatorname{tw}(\Gamma_n)$ to grow with $n$. For
example, if $\Gamma_n$ becomes a forest after deleting $s_n$ vertices or
$r_n$ edges, then $\operatorname{tw}(\Gamma_n)\le s_n+1$ or
$\operatorname{tw}(\Gamma_n)\le r_n+1$, respectively. Subject to the low-degree condition, \cref{thm:main} therefore applies whenever
$s_n=o(D_n/\log n)$ or $r_n=o(D_n/\log n)$, that is, to planted graphs that
are tree-like apart from a slowly growing number of exceptional vertices or
extra edges.

At and below the Erd\H{o}s--R\'enyi critical point, we can drop the treewidth condition altogether.

\begin{theorem}[Critical and Subcritical Regimes]
    Fix $c\in(0,1]$, and let $\{\Gamma_n\}_{n\ge1}$ be any sequence of
    simple graphs with at most $n$ vertices. Suppose that
    $D_n=\omega(\log n)$ and \eqref{eq:low-degree-condition} holds.
    Then, for every fixed $\varepsilon\in(0,1)$, the conclusion
    \eqref{eq:TV-consequence} holds.
    \label{thm:subcritical}
\end{theorem}

In other words, when $c\le1$, every deterministic sequence of planted graphs
that satisfies the low-degree condition at some degree $D_n=\omega(\log n)$
becomes statistically indistinguishable from the null after any constant
amount of noise, with no assumption on its treewidth.

\begin{remark}[Necessity of a Structural Condition]
Above criticality, some structural restriction on the planted graph is
necessary. For every fixed $\varepsilon\in(0,1)$ and all sufficiently large
$c>1$, \cref{thm:deterministic-gap} constructs a sequence of planted models
whose degree-$(\log n)^2$ advantages vanish, yet whose noised laws remain
statistically distinguishable from $Q_n$. We do not know whether the
treewidth condition in \cref{thm:main} is sharp, but this example shows that
no statement of this form can hold without some control on the structure of
the planted graph.
\end{remark}

\subsection{Related Work}
\label{subsec:related-work}

\paragraph{Comparison with \cite{HsiehEtAl2026}.}
The conclusions of \cref{thm:main,thm:subcritical} do not follow from the
Boolean-vector theorem of \cite{HsiehEtAl2026}. A graph can be encoded by
\(N=\binom n2\) edge indicators, but their theorem requires invariance under
the full coordinate-permutation group \(S_N\) and a product Bernoulli null
with fixed bias. Our model is invariant only under the action of \(S_n\) by
vertex relabeling, and its null bias \(c/n\) tends to zero. In place of full
symmetry, we use the monotone union structure of the model, which makes every
Fourier coefficient nonnegative and nonzero only on copies of subgraphs of
$\Gamma_n$, together with decompositions along small vertex interfaces.
Treewidth provides these decompositions above criticality. At and below
criticality, the low-degree condition excludes short cycles, noise destroys
the long ones with high probability, and the same decompositions apply to the
remaining forests. Our results also complement their Gaussian-matrix
theorem: we control the entire noised distribution of a sparse Bernoulli
graph, whereas they control each fixed connected signed subgraph count.

\vspace{-0.1in}
\paragraph{Treewidth and Color Coding.}
Treewidth also governs algorithms for evaluating subgraph statistics. The optimal degree-$D$ statistic $L_n^{\le D}$ is a linear
combination of signed counts of subgraphs $F\subseteq\Gamma_n$ with at most
$D$ edges, and each such $F$ satisfies
$\operatorname{tw}(F)\le\operatorname{tw}(\Gamma_n)$. Color
coding~\cite{AlonYusterZwick1995} finds a copy of a pattern with $s$
vertices and treewidth $k$ in time $2^{O(s)}n^{k+1}$ by a dynamic program
over its tree decomposition. The same dynamic program counts colorful copies,
and averaging over random colorings gives unbiased estimates of the count, as in the self-avoiding-walk
estimators of~\cite{HopkinsSteurer2017}; see~\cite{ArvindRaman2002} for
approximate counting of bounded-treewidth patterns and~\cite{AlonGutner2010}
for a deterministic version. Turning these per-pattern estimates into an accuracy guarantee for the full
low-degree statistic requires a variance analysis, which we do not pursue here. The dynamic program stores tables indexed by
the images of the bag intersections $B_x\cap B_y$, and our proof pays a
factor of $n$ for each vertex in the same intersections
(\cref{section:overview}).

\vspace{-0.1in}
\paragraph{Bounded Independence plus Noise.}
A related principle appears in pseudorandomness. A distribution with bounded
independence matches the uniform distribution in every moment up to a fixed
degree, which is an exact form of low-degree indistinguishability. Haramaty,
Lee, and Viola~\cite{HaramatyLeeViola2018} showed that such a distribution,
after independent noise is added, fools products of bounded functions on
disjoint blocks of input bits. Lee and Viola~\cite{LeeViola2020} extended this
to read-once polynomials, and Forbes and Kelley~\cite{ForbesKelley2018}
extended the analysis to read-once branching programs that read their input
in an unknown order. These results establish indistinguishability against a
structured class of tests. We instead obtain indistinguishability against all
tests, and we place the structural assumption on the planted graph.

\vspace{-0.1in}
\paragraph{Relation to Planted-Subgraph Inference.}
Adding a randomly placed copy of a fixed graph is a classical way to plant
structure, and planted clique is the best-known example
\cite{Jerrum1992,AlonKrivelevichSudakov1998}. We focus on sparse backgrounds
of density $c/n$, where this model has been studied for many planted
structures. Massouli\'e, Stephan, and Towsley study the same model in
$G(n,c/n)$ for planted paths, stars, and complete trees of fixed arity. They
obtain detection and recovery thresholds and, in their impossibility
regimes, second-moment bounds for the full likelihood ratio
\cite{MassoulieStephanTowsley2019}.

When the background edge density is $\Omega(1)$, Yu, Zadik, and Zhang
\cite{YuZadikZhang2025} show that signed star counts are optimal among
fixed-degree polynomial tests for every planted graph. In a more general
model, where each edge of an embedded template is present with probability
$p_n>q_n$, Elimelech and Huleihel characterize the statistical and computational
detection thresholds in the dense regime in terms of the number of edges, the
maximum degree, and the maximum subgraph density of the template, and also
analyze the sparse and critical regimes \cite{ElimelechHuleihel2025}. Their
semi-random extension lets a monotone adversary delete non-planted edges and
identifies fragile and robust detection regimes, governed by the maximum
subgraph density \cite{ElimelechHuleihel2026SemiRandom}.

Other work studies recovery. Mossel, Niles-Weed, Sohn, Sun, and Zadik relate
all-or-nothing recovery transitions to generalized expectation thresholds
\cite{MosselEtAl2023}, and Lee, Pernice, Rajaraman, and Zadik derive a
variational formula for the limiting MMSE of arbitrary planted graphs under a
mild density condition \cite{LeePerniceRajaramanZadik2025}. Huleihel gives
sharp conditions for exact recovery of arbitrary templates in a dense
$p_n>q_n$ model, together with a polynomial-time recovery algorithm and
low-degree lower bounds \cite{Huleihel2026Recovery}. He also studies detection in a dense background
when only non-adaptive edge queries are allowed, and obtains lower and
upper bounds that match up to polylogarithmic factors for several template
families \cite{Huleihel2026Query}. Addario-Berry, Angel, Lugosi, R\'acz, and
Schramm determine the statistical thresholds for planted perfect matchings
and spanning trees in a model calibrated to remove the trivial edge-count
signal \cite{AddarioBerryEtAl2026}.

These works mainly seek detection or recovery thresholds as a function of the
planted structure. Our question is different and conditional: once a
low-degree calculation shows that the advantage vanishes, when does this
conclusion extend to the entire noised distribution? The degree-one term in the low-degree condition already forces $|E(\Gamma_n)|=o(\sqrt n)$, so the condition rules out the most immediate edge-counting signal. What remains is how the mass of
larger planted patterns is organized, and this is where the low-interface
structure provided by treewidth enters.

\section{Technical Overview}
\label{section:overview}
The low-degree condition \eqref{eq:low-degree-condition} controls only the Fourier contributions of subgraphs
with at most $D_n$ edges, whereas the full $L^2(Q_n)$ norm of the noised
likelihood ratio receives contributions from subgraphs of every size. We
therefore need to bound each large-support contribution by low-degree ones,
with a combinatorial loss small enough for the noise to absorb.

Our proofs combine three ideas. First, we compute the Fourier coefficients
that enter the $L^2$ norm and use an exact combinatorial identity: the
product of the relevant subgraph-count and automorphism factors equals the
number of corresponding isomorphism triples. This identity yields the
multiplicative estimate that relates the Fourier weight of a large support to
the weights of its pieces. Second, when the planted graph has small
treewidth, we use a refined tree decomposition to partition each large
support into pieces with at most $D_n$ edges. The low-degree condition controls the Fourier weights of these pieces.
Reassembling them costs a factor of $n$ for each shared vertex, while the
noise damps the Fourier weight of a support by a factor $(1-\varepsilon)^2$
per edge. Each piece has many edges but shares only
$O(\operatorname{tw}(\Gamma_n))$ vertices on average, so the damping
outweighs the cost. Third, in the
critical and subcritical regimes, we treat cycles separately: the low-degree condition excludes short cycles, and noise destroys the remaining long cycles
with high probability, which reduces the problem to forests.

The appendix uses a different argument. By the probabilistic method, we
construct a deterministic sequence of bounded-degree, locally sparse,
high-girth planted graphs whose low-degree Fourier mass vanishes, while a
scan statistic still detects the planted edges that survive the noise.

\subsection{Fourier Mass as Weighted Subgraph Counts}

For an abstract graph $F$ without isolated vertices, let
$\mathsf v(F)=|V(F)|$, $\mathsf e(F)=|E(F)|$, and let
$N_{\Gamma_n}(F)$ be the number of subgraphs of $\Gamma_n$ isomorphic to
$F$. The exact Fourier calculation gives the nonnegative weight
\[
    \alpha_n(F)
    =\left(\frac{1-c/n}{c/n}\right)^{\mathsf e(F)}
      \frac{N_{\Gamma_n}(F)^2|\operatorname{Aut}(F)|}
      {(n)_{\mathsf v(F)}}.
\]
Regrouping Fourier supports by isomorphism class yields
\begin{equation}
    \left\lVert L_n^{\le D_n}\right\rVert_{L^2(Q_n)}^2-1
    =\sum_{\substack{F:\,N_{\Gamma_n}(F)>0\\
             1\le\mathsf e(F)\le D_n}}\alpha_n(F),
    \label{eq:overview-low-degree-mass}
\end{equation}
whereas the full norm of the noised likelihood ratio is
\begin{equation}
    \left\lVert L_{n,\varepsilon}\right\rVert_{L^2(Q_n)}^2-1
    =\sum_{\substack{F:\,N_{\Gamma_n}(F)>0\\1\le\mathsf e(F)}}
      (1-\varepsilon)^{2\mathsf e(F)}\alpha_n(F).
    \label{eq:overview-noised-mass}
\end{equation}
By \eqref{eq:low-degree-condition}, the right-hand side of
\eqref{eq:overview-low-degree-mass} is $o(1)$. Thus the terms with
$\mathsf e(F)\le D_n$ in \eqref{eq:overview-noised-mass} are already
controlled, and it remains to bound the tail $\mathsf e(F)>D_n$.

The degree-one term already has a useful consequence. The summand indexed by
$F=K_2$ on the right-hand side of \eqref{eq:overview-low-degree-mass}
satisfies
\[
    \alpha_n(K_2)
    =(1+o(1))\frac{2|E(\Gamma_n)|^2}{cn}.
\]
Since all summands in \eqref{eq:overview-low-degree-mass} are nonnegative,
\eqref{eq:low-degree-condition} implies $\alpha_n(K_2)=o(1)$ and therefore
$|E(\Gamma_n)|=o(\sqrt n)$. Every graph $F$ with
$N_{\Gamma_n}(F)>0$ then satisfies
$\mathsf e(F)\le |E(\Gamma_n)|=o(\sqrt n)$ and
$\mathsf v(F)\le2\mathsf e(F)=o(\sqrt n)$. Hence the factors
$(1-c/n)^{\mathsf e(F)}$ and
$n^{\mathsf v(F)}/(n)_{\mathsf v(F)}$ are both $1+o(1)$ uniformly over the
whole expansion, and up to this factor we may work with the simpler weight
\[
    \beta_n(F)
    :=\frac{n^{\mathsf e(F)-\mathsf v(F)}}{c^{\mathsf e(F)}}
      N_{\Gamma_n}(F)^2|\operatorname{Aut}(F)|.
\]
The exact expansions are \eqref{eq:ldlr-2} and \eqref{eq:noise-lr}, and we
prove the comparison with $\beta_n$ and the degree-one bound in
\cref{lem:leading-term,lem:edge}.

\subsection{The Isomorphism-Triple Correspondence}
The factor
$N_{\Gamma_n}(F)^2|\operatorname{Aut}(F)|$ does not obviously factor over an
edge decomposition of $F$. We get around this by counting a single
combinatorial object that combines two copies of $F$ with an isomorphism
between them. An \emph{isomorphism triple} for $F$ is a triple
\[
    (F_1^\circ,F_2^\circ,\psi),
\]
where $F_1^\circ,F_2^\circ\subseteq\Gamma_n$ are concrete copies of $F$ and
$\psi:F_1^\circ\to F_2^\circ$ is a graph isomorphism. We write
$\mathrm{Tri}(F)$ for the set of all such triples. There are
$N_{\Gamma_n}(F)^2$ ordered pairs of concrete copies, and each pair admits
exactly $|\operatorname{Aut}(F)|$ isomorphisms, so
\[
    |\mathrm{Tri}(F)|
    =N_{\Gamma_n}(F)^2|\operatorname{Aut}(F)|.
\]
The combinatorial factor in $\beta_n(F)$ is therefore the size of a single
concrete set.

Now partition the edges of $F$ into subgraphs $T_1,\ldots,T_m$:
\[
    E(F)=E(T_1)\sqcup\cdots\sqcup E(T_m).
\]
Each $T_i$ consists of its edges and their endpoints, so the pieces are
edge-disjoint but may share vertices. Given a triple
$(F_1^\circ,F_2^\circ,\psi)\in\mathrm{Tri}(F)$, take the fixed compatible
decomposition $F_1^\circ=T_{1,1}^\circ\cup\cdots\cup T_{1,m}^\circ$ of the
first copy and transport it through $\psi$:
\[
    T_{2,i}^\circ:=\psi(T_{1,i}^\circ),
    \qquad
    \psi_i:=\psi|_{T_{1,i}^\circ}.
\]
Then
$(T_{1,i}^\circ,T_{2,i}^\circ,\psi_i)\in\mathrm{Tri}(T_i)$ for every $i$.
This defines a restriction map
\[
    \Delta_F:\mathrm{Tri}(F)
    \longrightarrow \prod_{i=1}^m\mathrm{Tri}(T_i).
\]

The map $\Delta_F$ is injective. Given a tuple in its image, the union of the
first coordinates is $F_1^\circ$ and the union of the second coordinates is
$F_2^\circ$. Since $F$ has no isolated vertices, every vertex lies in some
piece, so the local maps $\psi_i$, which are restrictions of the same $\psi$,
determine $\psi$ on all of $F_1^\circ$. Not every tuple in
$\prod_i\mathrm{Tri}(T_i)$ glues to a global triple, but this only makes the
codomain larger. Injectivity therefore gives
\[
    |\mathrm{Tri}(F)|
    \le \prod_{i=1}^m|\mathrm{Tri}(T_i)|.
\]
The point of the triple formulation is that we never compare the
automorphism group of $F$ with those of its pieces; we only restrict the
elements of one concrete set to elements of others.

When $F$ varies, the tuple of abstract pieces $(T_1,\ldots,T_m)$ does not
record how the pieces are glued together. We therefore keep the concrete
local triples in the summation. Their coordinatewise unions recover the two
global copies and hence the type $F$, so different gluings cannot collide.
This is the aggregate injection of \cref{lem:aggregate-injection}.

The only loss in the product bound comes from vertices shared by several
pieces. Define
\[
    \operatorname{ov}(F;T_1,\ldots,T_m)
    :=\sum_{i=1}^m\mathsf v(T_i)-\mathsf v(F).
\]
Since the edge sets partition $E(F)$, we have
$\sum_i\mathsf e(T_i)=\mathsf e(F)$, and substituting the triple count into
the definition of $\beta_n$ gives
\begin{equation}
\begin{aligned}
    \beta_n(F)
    &\le
    \frac{n^{\mathsf e(F)-\mathsf v(F)}}{c^{\mathsf e(F)}}
    \prod_{i=1}^m|\mathrm{Tri}(T_i)| \\
    &=n^{\operatorname{ov}(F;T_1,\ldots,T_m)}
      \prod_{i=1}^m\beta_n(T_i).
\end{aligned}
    \label{eq:overview-product-bound}
\end{equation}
A vertex that lies in $r$ pieces contributes $r-1$ to
$\operatorname{ov}(F;T_1,\ldots,T_m)$, so each extra appearance of a shared
vertex costs a factor of $n$. For example, if we cut a path into $m$
consecutive subpaths, adjacent pieces share an endpoint,
$\operatorname{ov}=m-1$, and the total loss is $n^{m-1}$. The pieces $T_i$
need not be connected. It remains to construct edge partitions whose pieces
lie in the low-degree range and whose overlap is small, and this is where
treewidth enters. We prove the triple count, the restriction map, and the
product bound in
\cref{lem:isomorphism-triples,lem:map-injection,lem:aggregate-injection}
and \eqref{eq:beta_overlap_bound}.

\subsection{Tree Decompositions and the Interface Budget}

The proof needs only the following low-interface criterion. At some
intermediate scale $\ell_n$, every relevant connected support should admit
an edge partition such that, for a constant $C$ and an interface parameter
$a_n$,
\[
    \ell_n\le\mathsf e(T_i)\le C\ell_n,
    \qquad
    \operatorname{ov}(F;T_1,\ldots,T_m)\le a_n(m-1),
\]
where $C\ell_n\le D_n$ and $a_n\log n=o(\ell_n)$. Bounded treewidth gives
this criterion with $C=3$ and $a_n=\operatorname{tw}(\Gamma_n)+1$: we show
that every connected graph $F$ of treewidth at most $k$ with at least $\ell$
edges has an edge partition satisfying
\begin{equation}
    \ell\le\mathsf e(T_i)\le3\ell,
    \qquad
    \operatorname{ov}(F;T_1,\ldots,T_m)
    \le(k+1)(m-1).
    \label{eq:overview-treewidth-division}
\end{equation}

\begin{figure}[!t]
    \centering
    \setlength{\abovecaptionskip}{3pt}
    \setlength{\belowcaptionskip}{0pt}
    \begin{minipage}[c]{0.36\textwidth}
        \centering
        \begin{adjustbox}{max width=\linewidth,max totalheight=6.0cm}
            \input{figures/tree_decomposition_original_graph.tex}
        \end{adjustbox}
    \end{minipage}\hfill
    \begin{minipage}[c]{0.62\textwidth}
        \centering
        \begin{adjustbox}{max width=\linewidth,max totalheight=6.0cm}
            \input{figures/tree_decomposition_stage_one.tex}
        \end{adjustbox}
    \end{minipage}

    \begin{minipage}[t]{0.36\textwidth}
        \centering
        \caption{\textbf{The running graph $F$.} It has $17$ edges and treewidth $3$.}
        \label{fig:tree-decomposition-original-graph}
    \end{minipage}\hfill
    \begin{minipage}[t]{0.62\textwidth}
        \centering
        \caption{\textbf{Initial decomposition and edge ownership.} A node of
        the decomposition tree may own several edges; every bag has size at
        most $4$.}
        \label{fig:tree-decomposition-stage-one}
    \end{minipage}
    \vspace{1em}

    \begin{minipage}[c]{0.54\textwidth}
        \centering
        \begin{adjustbox}{max width=\linewidth,max totalheight=8.0cm}
            \input{figures/tree_decomposition_stage_two.tex}
        \end{adjustbox}
    \end{minipage}\hfill
    \begin{minipage}[c]{0.44\textwidth}
        \centering
        \begin{adjustbox}{max width=\linewidth,max totalheight=7.2cm}
            \input{figures/tree_decomposition_stage_three.tex}
        \end{adjustbox}
    \end{minipage}

    \begin{minipage}[t]{0.54\textwidth}
        \centering
        \caption{\textbf{Subcubic refinement and clustering.} Each node $x$ becomes a
        subcubic tree $R_x$ (dashed box) whose nodes carry the bag $B_x$; a box
        labeled by an edge owns that edge, the port $p_y$ attaches $R_x$ to
        $R_y$, and the two dotted edges are the cuts.}
        \label{fig:tree-decomposition-stage-two}
    \end{minipage}\hfill
    \begin{minipage}[t]{0.44\textwidth}
        \centering
        \caption{\textbf{The induced edge partition.} Here
        $\mathrm{ov}=6\le(k+1)(m-1)=8$, and the piece $T_2$ is disconnected.}
        \label{fig:tree-decomposition-stage-three}
    \end{minipage}
\end{figure}

To construct the partition, we first refine a width-$k$ tree decomposition
so that the decomposition tree is subcubic and each node owns at most one
edge of $F$. Weighting each node by the number of edges it owns turns the
decomposition tree into a tree with $0$--$1$ weights. A greedy procedure
splits this tree into connected clusters of total weight between $\ell$ and
$3\ell$, and the edges owned by a cluster form the corresponding piece $T_i$.
Contracting the clusters gives a quotient tree. Each edge of the quotient
tree comes from a unique edge $xy$ of the decomposition tree, and only the at
most $k+1$ vertices of $B_x\cap B_y$ can be counted on both sides of it. This
gives the overlap bound in \eqref{eq:overview-treewidth-division}. The same
bag intersections index the tables of the color-coding dynamic program in
\cref{subsec:related-work}. At an edge $xy$ of the decomposition tree, that
table has $n^{|B_x\cap B_y|}$ entries, which matches the factor
$n^{|B_x\cap B_y|}$ that cutting at $xy$ costs in our bound.
\Cref{fig:tree-decomposition-original-graph,fig:tree-decomposition-stage-one,fig:tree-decomposition-stage-two,fig:tree-decomposition-stage-three}
show a running example and the three stages of the construction. Here
$\ell=5$: the blue, orange, and green nodes form the three clusters, the
dotted edges mark the two cuts, and transporting edge ownership back to $F$
gives pieces with $5$, $7$, and $5$ edges.

Now choose an intermediate scale $\ell_n$ with
\[
    3\ell_n\le D_n,
    \qquad
    (\operatorname{tw}(\Gamma_n)+1)\log n=o(\ell_n).
\]
Then all pieces lie in the low-degree range. Since $\mathsf e(F)\ge m\ell_n$,
the noise factor times the interface loss in
\eqref{eq:overview-product-bound} is at most
\[
    (1-\varepsilon)^{2\mathsf e(F)}
    n^{\operatorname{ov}(F;T_1,\ldots,T_m)}
    \le
    \exp\!\left(
        2m\ell_n\log(1-\varepsilon)
        +(\operatorname{tw}(\Gamma_n)+1)(m-1)\log n
    \right).
\]
For every fixed $\varepsilon>0$ and all sufficiently large $n$, this is at
most one, so the noise pays for the entire interface loss. The aggregate
injection then bounds the sum over all large connected $F$ by a geometric
series in the low-degree mass, and a component expansion over connected
components handles disconnected $F$. We state the decomposition in
\cref{lem:decomposition} and prove it through
\cref{prop:treewidth-division,lem:decomposition-refinement,lem:partition-tree};
we bound the sum over connected graphs in \cref{lem:connect-upper-bound}
and pass to disconnected graphs in \cref{lem:cluster-expansion}.

\subsection{Cycles at and below Criticality}

For \cref{thm:subcritical}, a cycle argument replaces the treewidth
assumption. If $\Gamma_n$ contains a cycle $C_\ell$ with
$\ell\le D_n$, then
\[
    \beta_n(C_\ell)
    =\frac{2\ell}{c^\ell}N_{\Gamma_n}(C_\ell)^2
    \ge2\ell
    \qquad (c\le1).
\]
By \cref{lem:edge} and \eqref{eq:leading-term}, for all sufficiently large
$n$ this implies
\[
    \alpha_n(C_\ell)\ge\frac12\beta_n(C_\ell)\ge\ell,
\]
which contradicts \eqref{eq:low-degree-condition}, since
$\alpha_n(C_\ell)$ is a nonnegative summand in
\eqref{eq:overview-low-degree-mass}. Hence $\Gamma_n$ must have girth larger
than $D_n$.

It remains to show that long cycles do not survive the noise. Set
$r_n=\lfloor(D_n-1)/2\rfloor$. In a graph of girth greater than $D_n$, any two
vertices are joined by at most one path of length at most $r_n$. Recording
every $r_n$-th vertex of a $k$-cycle therefore determines the cycle from at
most $k/r_n+1$ vertices, which gives
\[
    |\mathcal C_k(\Gamma_n)|\le n^{k/r_n+1}.
\]
Each such cycle survives the noise with probability $(1-\varepsilon)^k$.
Since $D_n=\omega(\log n)$, summing over $k>D_n$ shows that the noised planted
graph is acyclic with probability $1-o(1)$.

We then pass to the forest-truncated model, in which the planted graph is
replaced by the empty graph whenever it contains a cycle, and write
$L_{n,\varepsilon}^{\mathrm F}$ for its likelihood ratio. Its Fourier
coefficients vanish on cyclic supports. For acyclic $F$, the coupling
$T_\varepsilon^{\mathrm F}\Gamma_n\subseteq T_\varepsilon\Gamma_n$
shows that the containment probability, and hence the Fourier coefficient,
is no larger than in the original noised model. By Parseval, the truncated
$L^2$ norm is therefore bounded by the forest part of \eqref{eq:noise-lr}. A
component expansion reduces this forest sum to the tree sum in
\eqref{eq:truncated-sum-goal}, and the decomposition argument above applies
with treewidth one. The short- and long-cycle estimates are
\cref{lem:girth,lem:no-cycle}, and the forest truncation is carried out in
\eqref{eq:truncated-tv-goal} and \eqref{eq:truncated-sum-goal}.

\subsection{The Deterministic Supercritical Example}
\label{subsec:overview-supercritical-gap}

The example in \cref{thm:deterministic-gap} is a deterministic sequence of
planted graphs, even though we prove its existence by the probabilistic
method. Fix $m_n=\lfloor n^r\rfloor$ with $r<1/2$ and an even integer $D$
with $(1-\varepsilon)D/2>1$. We take a realization of the $D$-regular configuration model on $m_n$
vertices that is simple~\cite{Bollobas1980RegularGraphs}, locally sparse, and
has at most $m_n^{3/4}$ cycles of length at most $h_n=\Theta(\log n)$, and we
delete one edge from each such cycle. The
resulting graph $\Gamma_n$ has maximum degree at most $D$, girth
$\Omega(\log n)$, and, for constants $A>0$ and $\gamma\in(0,1)$,
$\mathsf e(\Gamma_n[U])\le(1+A/\log n)|U|$ for $2\le|U|\le m_n^\gamma$, while keeping $\mathsf e(\Gamma_n)=(D/2-o(1))m_n$
edges; see \cref{lem:deterministic-planted-graphs}.

For $D_n=\lceil(\log n)^2\rceil$, we first sum the low-degree Fourier mass
over connected supports. By \cref{lem:deterministic-connected-count}, the
factors $N_{\Gamma_n}(F)^2|\operatorname{Aut}(F)|$ summed over connected $F$
with $e$ edges are at most $m_n^2D^{3e}$. Trees carry the factor
$n^{\mathsf e(F)-\mathsf v(F)}=n^{-1}$, so for $c>D^3$ their contribution is at
most $(m_n^2/n)\sum_{e\ge1}(D^3/c)^e=o(1)$. For cyclic supports, local sparsity bounds
$\mathsf e(F)-\mathsf v(F)$ and logarithmic girth forces
$\mathsf e(F)\ge\kappa\log n$ for a constant $\kappa>0$, so for every sufficiently large fixed $c$
their total contribution is the tail of a geometric series and vanishes.
Factoring over connected components proves \eqref{eq:deterministic-low-degree}.
After noise, the planted vertex set still spans more than $\tau m_n$ planted
edges with high probability, for some $1<\tau<(1-\varepsilon)D/2$, whereas a
union bound shows that under $Q_n$ no set of $m_n$ vertices is this dense.
The resulting scan test proves \eqref{eq:deterministic-statistical}.

\section{Preliminaries}
\label{section:model}

\subsection{Notation and Conventions}

All logarithms are natural. Throughout, $n$ is large enough that $0<c/n<1$. For a positive integer $n$, write $[n]:=\{1,\ldots,n\}$, and let $\binom{[n]}{2}$ denote the set of unordered pairs of distinct elements of $[n]$. 

All graphs are finite and simple, and an \textit{abstract graph} is an isomorphism class of finite simple graphs without isolated vertices. 
This loses no generality, since adding or removing isolated vertices of the planted graph changes neither of the planted distributions below.
We distinguish an abstract graph from its particular realizations and use a representative $F$ when writing its vertex and edge sets.
For a graph $F$, write
\begin{equation*}
    \mathsf{v}(F):=|V(F)|,\qquad \mathsf{e}(F):=|E(F)|,\qquad\operatorname{tw}(F):=\text{the treewidth of }F.
\end{equation*}
For an abstract graph $F$, let $N_{\Gamma}(F)$ be the number of not-necessarily-induced subgraphs of $\Gamma$ isomorphic to $F$, where a subgraph is specified by its edge set together with its incident vertices, and let $\Aut(F)$ denote the automorphism group of $F$. 
We use the falling-factorial notation
\begin{equation*}
    (n)_v:=n(n-1)\cdots(n-v+1)\quad (v\ge 1),\qquad(n)_0:=1.
\end{equation*}
A superscript $\bullet$ denotes an \textit{embedded copy} on $[n]$: thus,
$F^\bullet$ is a labeled graph with $V(F^\bullet)\subseteq[n]$, $E(F^\bullet)\subseteq\binom{[n]}{2}$, and $F^\bullet\simeq F$. 
A superscript $\circ$ instead denotes a concrete copy inside the planted graph $\Gamma_n$; thus, $F^\circ\subseteq\Gamma_n$ and $F^\circ\simeq F$. 

The symbols $o(1)$, $O(\cdot)$, and $\omega(\cdot)$ refer to the limit $n\to\infty$. 
Unless stated otherwise, $c$ and $\varepsilon$ are fixed in this limit. 
For probability measures $\mu$ and $\nu$, we write
$\lVert\mu-\nu\rVert_{\mathrm{TV}}:=\sup_A |\mu(A)-\nu(A)|$.

\subsection{The Planted and Noised Models}
We now define the models of \cref{section:intro} formally.

Given a graph $\Gamma_n$ to be planted, we define the planted distribution $P_n$ as follows. Sample an injective map $\phi_n:V(\Gamma_n)\to [n]$
uniformly at random, independently of a graph $G\sim G(n,c/n)$, and define
\begin{equation*}
    \phi_n(E(\Gamma_n)):=\left\{\{\phi_n(u),\phi_n(v)\}:\{u,v\}\in E(\Gamma_n)\right\}.
\end{equation*}
Then $P_n$ is the law of the graph with vertex set $[n]$ and edge set $E(G)\cup \phi_n(E(\Gamma_n))$.

The noise operator $T_{\varepsilon}$ acts on a graph on $[n]$ independently for each pair $e\in\binom{[n]}{2}$: with probability $\varepsilon$ it resamples the indicator of $e$ from the Bernoulli distribution with parameter $c/n$, and otherwise it leaves the indicator unchanged. 
Since $T_{\varepsilon}$ preserves $G(n,c/n)$, applying it to $P_n$ leaves the Erd\H{o}s--R\'enyi background distributed as $G(n,c/n)$, while each planted edge survives as a planted edge independently with probability $1-\varepsilon$.
Equivalently, $T_{\varepsilon}$ acts directly on $\Gamma_n$ by deleting each edge independently with probability $\varepsilon$. Writing $T_{\varepsilon}\Gamma_n$ for the resulting random subgraph, $T_{\varepsilon}P_n$ is the law of the union of an independently embedded copy of $T_{\varepsilon}\Gamma_n$ and an independent sample from $G(n,c/n)$.

\subsection{Fourier Expansion}
We now define the low-degree likelihood ratio precisely and compute its squared $L^2(Q_n)$-norm $\left\lVert L_n^{\le D_n}\right\rVert^2_{L^2(Q_n)}$ explicitly.

The Fourier basis reduces this norm to a sum of squared Fourier coefficients. 
We compute the coefficients by conditioning on the random embedding of $\Gamma_n$, and then group the sum according to the isomorphism class of the underlying graphs. 
The same grouping gives the analogous formula for the noised model
and a convenient leading-order approximation to each summand.

Let $L^2(Q_n)$ be equipped with the inner product
\begin{equation*}
    \langle f,g\rangle_{Q_n}:=\mathbb{E}_{G \sim Q_n}[f(G)g(G)],
    \qquad \|f\|_{Q_n}^2:=\mathbb{E}_{G \sim Q_n}[f(G)^2].
\end{equation*}
For an integer $D\ge 0$, let $\mathbb{R}[X]_{\le D}$ denote the subspace of functions on $\{0,1\}^{\binom{[n]}{2}}$ that can be represented by multilinear polynomials of degree at most $D$ in the coordinate functions $(\omega(e))_{e\in\binom{[n]}{2}}$, and let
    \begin{equation*}
        \Pi_{\le D}^{Q_n}: L^2(Q_n)
    \longrightarrow \mathbb{R}[X]_{\le D}
    \end{equation*}
denote the orthogonal projection with respect to the inner product $\langle \cdot,\cdot\rangle_{Q_n}$.
Then, the degree-$D$ low-degree likelihood ratio is
\begin{equation*}
    L_n^{\le D}:=\Pi_{\le D}^{Q_n} L_n.
\end{equation*}

To describe this projection, we use the $Q_n$-Fourier basis.
For a configuration $\omega\in \{0,1\}^{\binom{[n]}{2}}$ and an edge $e\in\binom{[n]}{2}$, let
\begin{equation*}
    r_e(\omega):=\begin{cases}
        \sqrt{\frac{1-c/n}{c/n}}, &\text{if }\omega(e)=1,\\
        -\sqrt{\frac{c/n}{1-c/n}}, &\text{if }\omega(e)=0.
    \end{cases}
\end{equation*}
For each edge set  $S\subset \binom{[n]}{2}$, define \begin{equation*}
    \chi_S(\omega)=\prod_{e\in S}r_e(\omega),
\end{equation*}
with the convention $\chi_{\emptyset}\equiv 1$.
Then the collection $\{\chi_S:S\subset\binom{[n]}{2}\}$ forms an orthonormal basis of $L^2(Q_n)$; see \cite[Chapter 8]{MR3443800}.
Thus, for $f:\{0,1\}^{\binom{[n]}{2}}\to\mathbb{R}$ and $S\subset\binom{[n]}{2}$, define the \textit{Fourier coefficient}
\begin{equation*}
    \hat{f}(S):=\langle f,\chi_S\rangle_{Q_n}
    =\mathbb{E}_{Q_n}[f\chi_S].
\end{equation*}
The corresponding \textit{Fourier expansion} is
\begin{equation}
    f(\omega)=\sum\limits_{S\subset\binom{[n]}{2}}\hat{f}(S)\chi_S(\omega).
    \label{eq:fourier-expansion-1}
\end{equation}
Since $\chi_S$ is a multilinear polynomial of degree $|S|$, the orthogonal projection onto the subspace of polynomials of degree at most
$D$ is given by
\begin{equation*}
    \Pi_{\le D}^{Q_n} (f)= \sum\limits_{\substack{S\subset \binom{[n]}{2}\\|S|\le D}} \hat{f}(S) \chi_S.
\end{equation*}

By orthonormality, the squared norm is a sum of squared Fourier
coefficients; this is Parseval's identity:
\begin{equation*}
    \|f\|_{L^2(Q_n)}^2
    =\mathbb{E}_{Q_n}[f^2]
    =\sum\limits_{S\subset\binom{[n]}{2}}\left[\hat{f}(S)\right]^2.
\end{equation*}
Applying it to $L_n^{\le D_n}=\Pi_{\le D_n}^{Q_n}(L_n)$ gives
\begin{equation}
        \left\lVert L_n^{\le D_n}\right\rVert_{L^2(Q_n)}^2
        =\mathbb{E}_{Q_n}\left[\Pi_{\le D_n}^{Q_n}(L_n)^2\right]
        =\sum\limits_{\substack{S\subset\binom{[n]}{2}\\ |S|\le D_n}}
        \left[\widehat{L_n}(S)\right]^2.
    \label{eq:ldlr-1}
\end{equation}

It remains to compute the coefficients $\widehat{L_n}(S)$ and to regroup the sum in \eqref{eq:ldlr-1}. 
An embedded copy $F^{\bullet}$ has no isolated vertices, so it is determined by its edge set,
and we use the shorthand
\begin{equation*}
    \chi_{F^\bullet}:=\chi_{E(F^\bullet)},
    \qquad
    \hat{f}(F^\bullet):=\hat{f}\bigl(E(F^\bullet)\bigr).
\end{equation*}

To compute these coefficients, we first condition on the planted copy. 
Fix an embedded copy $\Gamma_n^{\bullet}$.
Conditionally on $\phi_n(\Gamma_n)=\Gamma_n^\bullet$, the planted model forces $\omega(e)=1$ for every $e\in E(\Gamma_n^\bullet)$, 
so its likelihood ratio with respect to $Q_n$ is
\begin{equation*}
    L_{\Gamma_n^\bullet}(\omega):=\prod_{e\in E(\Gamma_n^\bullet)}\frac{\omega(e)}{c/n}.
\end{equation*}
For every edge $e$, 
\begin{equation*}
    \frac{\omega(e)}{c/n}=1+\sqrt{\frac{1-c/n}{c/n}}\,r_e(\omega).
\end{equation*}
Expanding the product over subsets of $E(\Gamma_n^\bullet)$ gives
\begin{equation}
        L_{\Gamma_n^\bullet}(\omega)
        =\prod_{e\in E(\Gamma_n^\bullet)}
        \left(1+\sqrt{\frac{1-c/n}{c/n}}\,r_e(\omega)\right)
        =\sum\limits_{F^\bullet\subset\Gamma_n^\bullet}
        \left(\frac{1-c/n}{c/n}\right)^{|E(F^\bullet)|/2}
        \chi_{F^\bullet}(\omega),
    \label{eq:conditioned-likelihood}
\end{equation}
where the summation is over the empty graph and every subgraph is specified by a subset of $E(\Gamma_n^\bullet)$ together with its incident vertices.

Since $P_n$ is the mixture of these conditional models over the uniformly random embedding of $\Gamma_n$, 
its likelihood ratio is the average of the conditional likelihood ratios. 
Taking expectations in
\eqref{eq:conditioned-likelihood}, we obtain
\begin{equation}
        L_n(\omega)
        =\mathbb{E}_{\Gamma_n^\bullet}
        \left[L_{\Gamma_n^\bullet}(\omega)\right]
        =\sum\limits_{F^\bullet}
        \mathbb{P}\left[F^\bullet\subset\phi_n(\Gamma_n)\right]
        \left(\frac{1-c/n}{c/n}\right)^{|E(F^\bullet)|/2}
        \chi_{F^\bullet}(\omega),
    \label{eq:likelihood-expansion}
\end{equation}
where the sum ranges over the empty graph and all embedded graphs on $[n]$ without isolated vertices. 
Comparing \eqref{eq:likelihood-expansion} with the Fourier expansion \eqref{eq:fourier-expansion-1} yields
\begin{equation*}
    \widehat{L_n}(F^\bullet)
    =\mathbb{P}\left[F^\bullet\subset\phi_n(\Gamma_n)\right]
    \left(\frac{1-c/n}{c/n}\right)^{|E(F^\bullet)|/2}.
\end{equation*}

We next express this containment probability, and hence the Fourier coefficient, in terms of the underlying abstract graph. 

Fix a nonempty abstract graph $F$ and an embedded copy $F^\bullet\simeq F$.
The falling factorial $(n)_{\mathsf{v}(F)}$ counts the injective maps from $V(F)$ to $[n]$. 
Each embedded copy of $F$ on $[n]$ arises from exactly
$|\mathrm{Aut}(F)|$ such maps. 
Hence, the total number of embedded copies is
\begin{equation*}
    \frac{(n)_{\mathsf{v}(F)}}{|\Aut(F)|}.
\end{equation*}
Under the random embedding, each of the $N_{\Gamma_n}(F)$ concrete copies of $F$ in $\Gamma_n$ is mapped to a uniformly random embedded copy, 
and since the embedding is injective, distinct subgraphs of $\Gamma_n$ have distinct images. 
Hence
\begin{equation*}
    \mathbb{P}\left[F^\bullet\subset\phi_n(\Gamma_n)\right]
    =\frac{N_{\Gamma_n}(F)|\mathrm{Aut}(F)|}
    {(n)_{\mathsf{v}(F)}}.
\end{equation*}
All embedded copies $F^\bullet\simeq F$ therefore have the same Fourier
coefficient, and hence
\begin{equation*}
    \begin{aligned}
        \sum\limits_{F^\bullet\simeq F}
        \left[\widehat{L_n}(F^\bullet)\right]^2
        &=\frac{(n)_{\mathsf{v}(F)}}{|\mathrm{Aut}(F)|}
        \left(
        \frac{N_{\Gamma_n}(F)|\mathrm{Aut}(F)|}
        {(n)_{\mathsf{v}(F)}}
        \right)^2
        \left(\frac{1-c/n}{c/n}\right)^{\mathsf{e}(F)}\\
        &=\left(\frac{1-c/n}{c/n}\right)^{\mathsf{e}(F)}
        \frac{N_{\Gamma_n}(F)^2|\mathrm{Aut}(F)|}
        {(n)_{\mathsf{v}(F)}}.
    \end{aligned}
\end{equation*}

We now partition the sum in \eqref{eq:ldlr-1} according to the
isomorphism class of the graph with edge set $S$. The empty graph
contributes $1$, since
$\widehat{L_n}(\emptyset)=\mathbb{E}_{Q_n}[L_n]=1$. Combining the
remaining terms with the preceding identity yields
\begin{equation}
    \left\lVert L_n^{\le D_n}\right\rVert_{L^2(Q_n)}^2
    =1+\sum\limits_{\substack{F:\,N_{\Gamma_n}(F)>0\\
    1\le\mathsf{e}(F)\le D_n}}
    \left(\frac{1-c/n}{c/n}\right)^{\mathsf{e}(F)}
    \frac{N_{\Gamma_n}(F)^2|\mathrm{Aut}(F)|}
    {(n)_{\mathsf{v}(F)}},
    \label{eq:ldlr-2}
\end{equation}
where the sum ranges over isomorphism classes of nonempty subgraphs of
$\Gamma_n$.

This is the exact expression for the low-degree second moment.
The same calculation applies to the noised model, in which each planted edge is retained independently with probability $1-\varepsilon$: 
for every nonempty $F$ and every embedded copy $F^\bullet\simeq F$,
\begin{equation*}
    \mathbb{P}\left[F^\bullet\subset \phi_n(T_\varepsilon\Gamma_n)\right]=(1-\varepsilon)^{\mathsf{e}(F)}\frac{N_{\Gamma_n}(F)|\mathrm{Aut}(F)|}{(n)_{\mathsf{v}(F)}}.
\end{equation*}
Consequently, for
$L_{n,\varepsilon}:=d(T_\varepsilon P_n)/dQ_n$, the same calculation as above gives
\begin{equation}
    \left\lVert L_{n,\varepsilon}\right\rVert_{L^2(Q_n)}^2=1+\sum\limits_{\substack{F:\,N_{\Gamma_n}(F)>0\\F\text{ nonempty}}}(1-\varepsilon)^{2\mathsf{e}(F)}\left(\frac{1-c/n}{c/n}\right)^{\mathsf{e}(F)}\frac{N_{\Gamma_n}(F)^2|\mathrm{Aut}(F)|}{(n)_{\mathsf{v}(F)}}.
    \label{eq:noise-lr}
\end{equation}

For later estimates, it is convenient to isolate the exact summand in \eqref{eq:ldlr-2} and compare it with a simpler expression. 
For each nonempty graph $F$ with
$N_{\Gamma_n}(F)>0$, define
\begin{equation*}
        \alpha_n(F)
        :=\left(\frac{1-c/n}{c/n}\right)^{\mathsf{e}(F)}
        \frac{N_{\Gamma_n}(F)^2|\mathrm{Aut}(F)|}
        {(n)_{\mathsf{v}(F)}},\qquad
        \beta_n(F)
        :=\frac{n^{\mathsf{e}(F)-\mathsf{v}(F)}}
        {c^{\mathsf{e}(F)}}
        N_{\Gamma_n}(F)^2|\mathrm{Aut}(F)|.
\end{equation*}

The next lemma shows that $\alpha_n(F)$ and $\beta_n(F)$ agree up to a factor $1+o(1)$, uniformly over any edge range $\mathsf{e}(F)=o(\sqrt{n})$.

\begin{lemma}
    Fix $c>0$ and let $d_n=o(\sqrt{n})$. Uniformly over all nonempty graphs $F$ without isolated vertices such that $N_{\Gamma_n}(F)>0$ and $\mathsf{e}(F)\le d_n$,
    \begin{equation*}
        \frac{\alpha_n(F)}{\beta_n(F)}=1+o(1).
    \end{equation*}
    Consequently, for all sufficiently large $n$,
    \begin{equation}
        \frac{1}{2}\beta_n(F)
        \le\alpha_n(F)
        \le 2\beta_n(F).
        \label{eq:leading-term}
    \end{equation}
    \label{lem:leading-term}
\end{lemma}

\begin{proof}
    Direct cancellation gives
    \begin{equation*}
        \frac{\alpha_n(F)}{\beta_n(F)}
        =\left(1-\frac{c}{n}\right)^{\mathsf{e}(F)}
        \frac{n^{\mathsf{v}(F)}}{(n)_{\mathsf{v}(F)}}
        =\left(1-\frac{c}{n}\right)^{\mathsf{e}(F)}
        \prod_{j=0}^{\mathsf{v}(F)-1}
        \left(1-\frac{j}{n}\right)^{-1}.
    \end{equation*}
    Since $F$ has no isolated vertices,
    $\mathsf{v}(F)\le 2\mathsf{e}(F)\le 2d_n=o(\sqrt{n})$. 
    For all sufficiently large $n$, we may therefore use
    $|\log(1-x)|\le 2x$ for $0\le x\le 1/2$ to obtain
    \begin{align*}
        \left|\log\frac{\alpha_n(F)}{\beta_n(F)}\right|
        &\le \frac{2c\mathsf{e}(F)}{n}
        +\frac{2}{n}\sum_{j=0}^{\mathsf{v}(F)-1}j\\
        &\le \frac{2c\mathsf{e}(F)}{n}
        +\frac{\mathsf{v}(F)^2}{n}
        =O\left(\frac{d_n}{n}+\frac{d_n^2}{n}\right)
        =o(1).
    \end{align*}
    Exponentiating proves the uniform ratio estimate and hence
    \eqref{eq:leading-term}.
\end{proof}
\subsection{Direct Consequences of Low-Degree Indistinguishability}

We now derive two properties of $\Gamma_n$ directly from \eqref{eq:low-degree-condition}. The first bounds the number of edges of $\Gamma_n$.

\begin{lemma}
    Suppose that $D_n\ge 1$ for all sufficiently large $n$ and that
    \eqref{eq:low-degree-condition} holds.
    Then
    \begin{equation*}
        \mathsf{e}(\Gamma_n)=o(\sqrt{n}).
    \end{equation*}
    Consequently, the bounds in \eqref{eq:leading-term} hold uniformly over all nonempty graphs $F$ with $N_{\Gamma_n}(F)>0$ and no isolated vertices.
    \label{lem:edge}
\end{lemma}

\begin{proof}    
    If $\mathsf e(\Gamma_n)=0$, there is nothing to prove. Otherwise,
    $N_{\Gamma_n}(K_2)>0$, and since $D_n\ge1$ for all sufficiently large $n$, the term
    corresponding to $K_2$ appears in \eqref{eq:ldlr-2}.
    All summands in \eqref{eq:ldlr-2} are nonnegative, so \eqref{eq:low-degree-condition} implies $\alpha_n(K_2)=o(1)$.
    Since $\mathsf{e}(K_2)=1$, Lemma \ref{lem:leading-term}, applied with $d_n\equiv 1$, gives
    \begin{equation*}
        \beta_n(K_2)\le 2\alpha_n(K_2)=o(1).
    \end{equation*}
    Moreover, $|\mathrm{Aut}(K_2)|=2$ and $N_{\Gamma_n}(K_2)=\mathsf{e}(\Gamma_n)$, and hence
    \begin{equation*}
        \beta_n(K_2)=\frac{2}{cn}\mathsf{e}(\Gamma_n)^2.
    \end{equation*}
    It follows that $\mathsf{e}(\Gamma_n)=o(\sqrt{n})$.

    Finally, set $d_n:=\max\{1,\mathsf{e}(\Gamma_n)\}=o(\sqrt{n})$. 
    Every nonempty abstract graph $F$ with $N_{\Gamma_n}(F)>0$ satisfies $\mathsf{e}(F)\le d_n$. 
    The uniform conclusion of Lemma \ref{lem:leading-term} therefore proves the second claim.
\end{proof}

For $0<c\le 1$, low-degree indistinguishability also rules out short cycles.

\begin{lemma}
    Fix $c\in(0,1]$. Suppose that $D_n=\omega(\log n)$ and that
    \eqref{eq:low-degree-condition} holds. 
    Then, for all sufficiently large $n$,
    \begin{equation*}
        N_{\Gamma_n}(C_\ell)=0
        \qquad\text{for every integer }3\le \ell\le D_n,
    \end{equation*}
    where $C_\ell$ denotes the cycle of length $\ell$.
    \label{lem:girth}
\end{lemma}

\begin{proof}
    Suppose, to the contrary, that for some integer
    $3\le \ell\le D_n$, the graph $\Gamma_n$ contains a cycle $C_\ell$. Then $N_{\Gamma_n}(C_\ell)\ge 1$, and Lemma \ref{lem:edge} gives
    \begin{equation*}
        \alpha_n(C_\ell)
        \ge\frac{1}{2}\beta_n(C_\ell).
    \end{equation*}
    Since $\mathsf{v}(C_\ell)=\mathsf{e}(C_\ell)=\ell$ and $|\mathrm{Aut}(C_\ell)|=2\ell$, the definition of $\beta_n$ yields
    \begin{equation*}
        \beta_n(C_\ell)=\frac{2\ell}{c^\ell}N_{\Gamma_n}(C_\ell)^2\ge 2\ell,
    \end{equation*}
    where the inequality uses $0<c\le 1$. 
    It follows from \eqref{eq:ldlr-2} that
    \begin{equation*}
        \left\lVert L_n^{\le D_n}\right\rVert_{L^2(Q_n)}^2-1 \ge\alpha_n(C_\ell)\ge \ell\ge 3,
    \end{equation*}
    contradicting \eqref{eq:low-degree-condition}. 
    Thus, for all sufficiently large $n$, no such cycle exists.
\end{proof}
\section{Isomorphism Triples and Proofs of the Main Results}
\label{section:proof}
In this section, we prove \cref{thm:main,thm:subcritical}. Comparing \eqref{eq:ldlr-2} with \eqref{eq:noise-lr}, we see that the main task is to control the contribution to \eqref{eq:noise-lr} of subgraphs with more than $D_n$ edges. 
We decompose each such subgraph into pieces in the low-degree range and bound its weight by the weights of the pieces.

\subsection{Isomorphism Triples and Low-Interface Decomposition}

Under \eqref{eq:low-degree-condition}, Lemmas
\ref{lem:leading-term} and \ref{lem:edge} imply that
\begin{equation}
    \sum\limits_{\substack{F:\,N_{\Gamma_n}(F)>0\\1\le\mathsf{e}(F)\le D_n}}
    \beta_n(F)=o(1).
    \label{eq:beta-low-degree}
\end{equation}
Moreover, by \eqref{eq:noise-lr} and Lemma \ref{lem:edge}, it suffices
for Theorem \ref{thm:main} to prove that
\begin{equation}
    \sum\limits_{\substack{F:\,N_{\Gamma_n}(F)>0\\
    1\le\mathsf{e}(F)}}
    (1-\varepsilon)^{2\mathsf{e}(F)}\beta_n(F)=o(1).
    \label{eq:beta-noised-target}
\end{equation}
We first reduce to connected $F$ by factoring over components
(\cref{lem:cluster-expansion}). Connected $F$ with fewer than $\ell_n$ edges,
for a scale $\ell_n$ with $3\ell_n\le D_n$ chosen in
\cref{lem:connect-upper-bound}, are controlled by \eqref{eq:beta-low-degree}.
For larger connected $F$, we decompose a concrete copy of $F$ in $\Gamma_n$
into subgraphs with between $\ell_n$ and $3\ell_n$ edges and bound
$\beta_n(F)$ by the weights of these subgraphs.
Define
\begin{equation*}
    \mathcal{C}_{\Gamma_n}(F):=\left\{F^\circ\subset\Gamma_n:F^\circ\simeq F\right\}.
\end{equation*}
Thus, $\mathcal{C}_{\Gamma_n}(F)$ is the collection of all subgraphs of $\Gamma_n$ isomorphic to $F$, and $|\mathcal{C}_{\Gamma_n}(F)|=N_{\Gamma_n}(F)$. 
For $F^\circ\in\mathcal{C}_{\Gamma_n}(F)$,
an \textit{ordered edge decomposition} of $F^\circ$ is an ordered tuple
\begin{equation}
    \mathcal{T}(F^\circ)=(T_1^\circ,T_2^\circ,\ldots,T_m^\circ)
    \label{eq:def-edge-decomposition}
\end{equation}
of nonempty subgraphs $T_i^\circ\subset F^\circ$ whose edge sets form an ordered partition of $E(F^\circ)$, namely,
\begin{equation*}
    E(F^\circ)=E(T_1^\circ)\sqcup E(T_2^\circ)\sqcup\cdots\sqcup E(T_m^\circ).
\end{equation*}
Each $T_i^\circ$ is specified by its edge set together with its incident vertices and hence has no isolated vertices.
We write $T_i$ for the isomorphism class of $T_i^\circ$. 
Passing \eqref{eq:def-edge-decomposition} to the abstract graph gives
\begin{equation*}
    \mathcal{T}(F)=(T_1,T_2,\ldots,T_m).
\end{equation*}
Whenever an ordered edge decomposition $\mathcal{T}(F)$ of a nonempty abstract graph $F$ with $N_{\Gamma_n}(F)>0$ is specified, an ordered edge decomposition of $F^\circ$ is called \textit{compatible} with $\mathcal{T}(F)$ if $T_i^\circ\in\mathcal{C}_{\Gamma_n}(T_i)$ for every $i=1,2,\ldots,m$.
For each $F^\circ\in\mathcal{C}_{\Gamma_n}(F)$, we fix one such compatible ordered edge decomposition $\mathcal{T}(F^\circ)$.
Such a decomposition exists: transport an ordered edge decomposition of a representative of $F$ along any isomorphism onto $F^\circ$. In a triple, only the decomposition of the first copy is fixed in advance; the decomposition of the second copy is its image under the isomorphism.
All maps and bounds below are relative to this specified decomposition.

With these decompositions fixed, we next interpret $\beta_n(F)$ as a count, which will let us compare it with $\beta_n(T_1),\ldots,\beta_n(T_m)$.
For $F_1^\circ,F_2^\circ\in\mathcal{C}_{\Gamma_n}(F)$, let
$\psi:V(F_1^\circ)\to V(F_2^\circ)$ be a graph isomorphism; that is, $\psi$ is a bijection satisfying
\begin{equation*}
    \{u,v\}\in E(F_1^\circ)
    \quad\Longleftrightarrow\quad
    \{\psi(u),\psi(v)\}\in E(F_2^\circ).
\end{equation*}
We then define the set of \textit{isomorphism triples} associated with the abstract graph $F$ by
\begin{equation*}
    \mathrm{Tri}(F)
    :=\left\{
        (F_1^\circ,F_2^\circ,\psi):
        F_1^\circ,F_2^\circ\in\mathcal{C}_{\Gamma_n}(F),
        \psi:F_1^\circ\xrightarrow{\simeq}F_2^\circ
    \right\}.
\end{equation*}
\begin{lemma}
    For every nonempty abstract graph $F$ with $N_{\Gamma_n}(F)>0$,
    \begin{equation*}
        |\mathrm{Tri}(F)|=N_{\Gamma_n}(F)^2|\mathrm{Aut}(F)|.
    \end{equation*}
    Furthermore, we have 
    \begin{equation*}
        \beta_n(F)=\frac{n^{\mathsf{e}(F)-\mathsf{v}(F)}}{c^{\mathsf{e}(F)}}|\mathrm{Tri}(F)|.
    \end{equation*}
    \label{lem:isomorphism-triples}
\end{lemma}

\begin{proof}
    There are $N_{\Gamma_n}(F)^2$ ordered pairs $(F_1^\circ,F_2^\circ)$ in $\mathcal{C}_{\Gamma_n}(F)\times \mathcal{C}_{\Gamma_n}(F)$. 
    For each such pair, the set of graph isomorphisms from $F_1^\circ$ to $F_2^\circ$ is in bijection with $\mathrm{Aut}(F)$ and therefore contains exactly $|\mathrm{Aut}(F)|$ maps. 
    This proves the first identity. The second follows immediately from the definition of $\beta_n(F)$.
\end{proof}

We now restrict triples to the pieces of the specified decomposition. 
Let $(F_1^\circ,F_2^\circ,\psi)\in\mathrm{Tri}(F)$, and let
\begin{equation*}
    \mathcal{T}(F_1^\circ)
    =\left(T_{1,1}^\circ,T_{1,2}^\circ,\ldots,T_{1,m}^\circ\right)
\end{equation*}
be the chosen compatible ordered edge decomposition of $F_1^\circ$. 
For $i=1,2,\ldots,m$, let
\begin{equation*}
    T_{2,i}^\circ:=\psi(T_{1,i}^\circ),
    \qquad
    \psi_i=\psi\mid_{T_{1,i}^\circ}.
\end{equation*}
Then $(T_{1,i}^\circ,T_{2,i}^\circ,\psi_i)\in\mathrm{Tri}(T_i)$.
These triples define the following map.
\begin{equation}
    \begin{aligned}
        \Delta_F:& \mathrm{Tri}(F)
        &\qquad \longrightarrow \qquad &
        \mathrm{Tri}(T_1)\times\mathrm{Tri}(T_2)\times\cdots
        \times\mathrm{Tri}(T_m),\\  &
        (F_1^\circ,F_2^\circ,\psi)
        &\qquad \longmapsto \qquad &
        \left(
            (T_{1,1}^\circ,T_{2,1}^\circ,\psi_1),
            (T_{1,2}^\circ,T_{2,2}^\circ,\psi_2),
            \cdots,
            (T_{1,m}^\circ,T_{2,m}^\circ,\psi_m)
        \right).
    \end{aligned}
    \label{eq:isomorphism-map}
\end{equation}

The following lemma gives the corresponding product bound for the triple counts.

\begin{lemma}
    For every nonempty abstract graph $F$ with $N_{\Gamma_n}(F)>0$, the map $\Delta_F$ defined in \eqref{eq:isomorphism-map} is injective.
    Consequently, we have 
    \begin{equation*}
        |\mathrm{Tri}(F)|\le \prod_{i=1}^m |\mathrm{Tri}(T_i)|.
    \end{equation*}
    \label{lem:map-injection}
\end{lemma}

\begin{proof}
    Consider any tuple
    \begin{equation*}
        \left(
            (T_{1,1}^\circ,T_{2,1}^\circ,\psi_1),
            (T_{1,2}^\circ,T_{2,2}^\circ,\psi_2),
            \cdots,
            (T_{1,m}^\circ,T_{2,m}^\circ,\psi_m)
        \right)
        \in\Delta_F\left(\mathrm{Tri}(F)\right).
    \end{equation*}
    Since the edge sets of the first coordinates form an ordered edge decomposition of $F_1^\circ$, the first coordinates recover
    \begin{equation*}
        F_1^\circ=\bigcup_{i=1}^m T_{1,i}^\circ.
    \end{equation*}
    Similarly, since $F_2^\circ=\psi(F_1^\circ)$, the second coordinates recover
    \begin{equation*}
        F_2^\circ=\bigcup_{i=1}^m T_{2,i}^\circ.
    \end{equation*}
    Moreover, because $F_1^\circ$ has no isolated vertices, every vertex $v\in V(F_1^\circ)$ belongs to $V(T_{1,j}^\circ)$ for some $j\in\{1,2,\ldots,m\}$.
    Hence, $\psi(v)$ is determined by $\psi(v)=\psi_j(v)$. 
    Consequently, the tuple uniquely determines
    $(F_1^\circ,F_2^\circ,\psi)$, and therefore $\Delta_F$ is injective.
\end{proof}

Injectivity persists when we take the disjoint union over many abstract graphs at once.

\begin{lemma}
    Let $\mathcal F$ be a collection of nonempty abstract graphs, and fix
    $m\ge1$. For each $F\in\mathcal F$, specify an ordered edge
    decomposition
    \[
        \mathcal T(F)=(T_1(F),\ldots,T_m(F))
    \]
    and compatible decompositions of its concrete copies. Then the maps
    $\Delta_F$ combine to an injection
    \begin{equation*}
        \coprod_{F\in\mathcal F}\mathrm{Tri}(F)
        \longrightarrow
        \coprod_{(T_1,\ldots,T_m)}
        \prod_{i=1}^m\mathrm{Tri}(T_i),
    \end{equation*}
    where the disjoint union on the right is over the ordered tuples of
    piece types arising from the specified decompositions.
    \label{lem:aggregate-injection}
\end{lemma}

\begin{proof}
    A tuple in the image recovers the two global concrete graphs by taking
    the unions of its first and second coordinates. The local maps recover
    the global isomorphism because every vertex lies in at least one piece.
    Finally, the isomorphism type of the first global graph recovers $F$.
    Thus distinct elements of the disjoint-union domain cannot have the same
    image.
\end{proof}

For any specified ordered edge decomposition $\mathcal T(F)$ of a nonempty
abstract graph $F$ with $N_{\Gamma_n}(F)>0$, Lemmas
\ref{lem:isomorphism-triples} and \ref{lem:map-injection} together imply that
\begin{equation}
    \beta_n(F)\le n^{\mathrm{ov}(F,\mathcal{T}(F))}\prod_{i=1}^m \beta_n(T_i),
    \label{eq:beta_overlap_bound}
\end{equation}
where
\begin{equation*}
    \mathrm{ov}(F,\mathcal{T}(F)):=\sum\limits_{i=1}^m \mathsf{v}(T_i)-\mathsf{v}(F)
\end{equation*}
is the vertex-overlap count of the ordered edge decomposition $\mathcal{T}(F)$. 
To use this bound, we need decompositions whose pieces have controlled sizes and whose vertex-overlap count is small. 
Let $\ell,C,a$ be positive integers. 
We say that a graph $\Gamma$ is $(\ell,C,a)$-\textit{divisible} if, for every connected graph $F$ with $N_{\Gamma}(F)>0$ and $\mathsf{e}(F)\ge \ell$, there exists an ordered edge decomposition \begin{equation*}
    \mathcal{T}(F)=(T_1,T_2,\ldots,T_m)
\end{equation*}
such that
\begin{equation*}
    \ell\le\mathsf{e}(T_i)\le C\ell\quad\text{for every }i=1,2,\ldots,m,\qquad\mathrm{ov}(F,\mathcal{T}(F))\le a(m-1).
\end{equation*}
By the following lemma, bounded treewidth guarantees divisibility; we prove it in \cref{section:decomposition}.

\begin{lemma}
    Let $k,\ell$ be positive integers. Every graph $\Gamma$ with treewidth at most $k$ is $(\ell,3,k+1)$-divisible.
    \label{lem:decomposition}
\end{lemma}

\subsection{Bounded-Treewidth Planted Graphs}
We now prove \cref{thm:main}, assuming \cref{lem:decomposition}. 

We first use a component expansion to reduce the sum in \eqref{eq:beta-noised-target} to connected graphs.
Write
\begin{equation*}
    \theta_{n,\varepsilon}:=\sum\limits_{\substack{F:\,N_{\Gamma_n}(F)>0\\F\text{ connected}}}(1-\varepsilon)^{2\mathsf{e}(F)}\beta_n(F).
\end{equation*}

\begin{lemma}
    \label{lem:cluster-expansion}
    Suppose
    that $\theta_{n,\varepsilon}=o(1)$. 
    Then, for all sufficiently large $n$,
    \begin{equation*}
        \sum\limits_{\substack{F:\,N_{\Gamma_n}(F)>0\\1\le\mathsf{e}(F)}}(1-\varepsilon)^{2\mathsf{e}(F)}\beta_n(F)\le\frac{\theta_{n,\varepsilon}}{1-\theta_{n,\varepsilon}}
        =o(1).
    \end{equation*}
\end{lemma}

\begin{proof}
    Let $F$ be a nonempty abstract graph with $N_{\Gamma_n}(F)>0$, and let $T_1,T_2,\ldots,T_m$ be its connected components in an arbitrary order.
    These components form an ordered edge decomposition of $F$ with vertex-overlap count zero. 
    It therefore follows from \eqref{eq:beta_overlap_bound} that
    \begin{equation*}
        \beta_n(F)\le\prod_{i=1}^m\beta_n(T_i).
    \end{equation*}
    Since $\mathsf{e}(F)=\sum_{i=1}^m\mathsf{e}(T_i)$, we further obtain
    \begin{equation*}
        (1-\varepsilon)^{2\mathsf{e}(F)}\beta_n(F)\le\prod_{i=1}^m(1-\varepsilon)^{2\mathsf{e}(T_i)}\beta_n(T_i).
    \end{equation*}
    Every abstract graph with exactly $m$ connected components determines, after ordering its components, an ordered $m$-tuple of connected abstract graphs. 
    Hence, summing the above inequality and enlarging
    the resulting sum to all such ordered $m$-tuples gives
    \begin{equation*}
        \sum\limits_{\substack{F:\,N_{\Gamma_n}(F)>0\\F\text{ has exactly }m\text{ connected components}}}
        (1-\varepsilon)^{2\mathsf{e}(F)}\beta_n(F)
        \le\theta_{n,\varepsilon}^m.
    \end{equation*}
    Since $\theta_{n,\varepsilon}=o(1)$, we have $\theta_{n,\varepsilon}<1$ for all sufficiently large $n$. Therefore,
    \begin{equation*}
        \sum\limits_{\substack{F:\,N_{\Gamma_n}(F)>0\\
        1\le\mathsf{e}(F)}}(1-\varepsilon)^{2\mathsf{e}(F)}\beta_n(F)\le\sum_{m\ge 1}\theta_{n,\varepsilon}^m
        =\frac{\theta_{n,\varepsilon}}{1-\theta_{n,\varepsilon}}
        =o(1).
    \end{equation*}
\end{proof}

We now bound the contribution of connected graphs using low-interface decompositions.
Decomposing $F$ into smaller subgraphs creates the factor $n^{\mathrm{ov}(F,\mathcal{T}(F))}$ in \eqref{eq:beta_overlap_bound}. 
For the noise factor $(1-\varepsilon)^{2\mathsf{e}(F)}$ to dominate it, we need $\mathrm{ov}(F,\mathcal T(F))\log n=o(\mathsf{e}(F))$, which follows from \eqref{eq:tree-width-condition}.

\begin{lemma}
    Suppose that \eqref{eq:tree-width-condition} and
    \eqref{eq:low-degree-condition} hold. Then, for every fixed
    $\varepsilon\in(0,1)$,
    \begin{equation*}
        \theta_{n,\varepsilon}=o(1).
    \end{equation*}
    \label{lem:connect-upper-bound}
\end{lemma}

\begin{proof}
    If $\Gamma_n$ is edgeless, then $\theta_{n,\varepsilon}=0$.
    Along the remaining indices, $\operatorname{tw}(\Gamma_n)\ge1$, so
    \eqref{eq:tree-width-condition} implies
    \[
        (\operatorname{tw}(\Gamma_n)+1)\log n=o(D_n).
    \]
    We may therefore choose positive integers $\ell_n$ such that
    \begin{equation*}
        3\ell_n\le D_n,\qquad
        (\operatorname{tw}(\Gamma_n)+1)\log n=o(\ell_n).
    \end{equation*}
    Lemma \ref{lem:decomposition} then implies that $\Gamma_n$ is
    $(\ell_n,3,\operatorname{tw}(\Gamma_n)+1)$-divisible.

    We first consider connected abstract graphs with fewer than $\ell_n$ edges. 
    Since $(1-\varepsilon)^{2\mathsf{e}(F)}\le 1$ and $\ell_n\le D_n$, \eqref{eq:beta-low-degree} gives
    \begin{equation}
        \sum\limits_{\substack{F:\,N_{\Gamma_n}(F)>0,\ F\text{ connected}\\1\le\mathsf{e}(F)<\ell_n}}
        (1-\varepsilon)^{2\mathsf{e}(F)}\beta_n(F)
        \le
        \sum\limits_{\substack{F:\,N_{\Gamma_n}(F)>0\\1\le\mathsf{e}(F)\le D_n}}
        \beta_n(F)=o(1).
        \label{eq:upper-bound-small}
    \end{equation}

    Next, let $F$ be a connected abstract graph with $N_{\Gamma_n}(F)>0$ and $\mathsf{e}(F)\ge \ell_n$. By divisibility, $F$ admits an ordered edge decomposition
    \begin{equation*}
        \mathcal{T}(F)=(T_1,T_2,\ldots,T_m)
    \end{equation*}
    satisfying
    \begin{equation*}
        \ell_n\le\mathsf{e}(T_i)\le 3\ell_n
        \quad\text{for every }i=1,2,\ldots,m,\qquad
        \mathrm{ov}(F,\mathcal{T}(F))\le(\operatorname{tw}(\Gamma_n)+1)(m-1).
    \end{equation*}
    Since
    \begin{equation*}
        \mathsf{e}(F)=\sum_{i=1}^m\mathsf{e}(T_i)\ge m\ell_n,
    \end{equation*}
    \eqref{eq:beta_overlap_bound} yields
    \begin{align}
        (1-\varepsilon)^{2\mathsf{e}(F)}\beta_n(F)
        &\le (1-\varepsilon)^{2\mathsf{e}(F)}n^{(\operatorname{tw}(\Gamma_n)+1)(m-1)}\prod_{i=1}^m\beta_n(T_i)\nonumber\\
        &\le \exp\left(2m \ell_n\log(1-\varepsilon)+(\operatorname{tw}(\Gamma_n)+1)(m-1)\log n
        \right)\prod_{i=1}^m\beta_n(T_i).
        \label{eq:bounded-tw-single}
    \end{align}
    Because $(\operatorname{tw}(\Gamma_n)+1)\log n=o(\ell_n)$ and $\log(1-\varepsilon)<0$, the exponential factor is at most one, uniformly over $m\ge1$, for all sufficiently large $n$.

    Fix $m\ge1$, and let $\mathcal F_{n,m}$ be the collection of these
    connected graphs whose specified decompositions have exactly $m$ pieces.
    For an abstract graph $H$, set
    \[
        w_n(H):=n^{\mathsf e(H)-\mathsf v(H)}c^{-\mathsf e(H)}.
    \]
    Then $\beta_n(H)=w_n(H)|\mathrm{Tri}(H)|$, and for
    $F\in\mathcal F_{n,m}$,
    \[
        w_n(F)
        =n^{\mathrm{ov}(F,\mathcal T(F))}
          \prod_{i=1}^m w_n(T_i(F)).
    \]
    The estimate used in \eqref{eq:bounded-tw-single} gives
    \[
        (1-\varepsilon)^{2\mathsf e(F)}
        n^{\mathrm{ov}(F,\mathcal T(F))}\le1.
    \]
    Consequently,
    \begin{equation*}
        \sum_{F\in\mathcal F_{n,m}}
        (1-\varepsilon)^{2\mathsf e(F)}\beta_n(F)
        =
        \sum_{F\in\mathcal F_{n,m}}
        \sum_{\tau\in\mathrm{Tri}(F)}
        (1-\varepsilon)^{2\mathsf e(F)}w_n(F)
    \le
        \sum_{F\in\mathcal F_{n,m}}
        \sum_{\tau\in\mathrm{Tri}(F)}
        \prod_{i=1}^m w_n(T_i(F)).
    \end{equation*}
    We do not sum \eqref{eq:bounded-tw-single} over $F$ directly, since distinct graphs $F$ can have the same tuple of piece types; instead, we sum over concrete triples.
    By \cref{lem:aggregate-injection}, distinct pairs $(F,\tau)$ give distinct
    ordered tuples of isomorphism triples. Enlarging the image to all
    ordered tuples of isomorphism triples whose types lie in the required edge
    range therefore gives
    \begin{equation*}
        \sum_{F\in\mathcal F_{n,m}}
        (1-\varepsilon)^{2\mathsf e(F)}\beta_n(F)
        \le
        \sum_{\substack{(T_1,\ldots,T_m):\\
                    N_{\Gamma_n}(T_i)>0,\\
                    \ell_n\le\mathsf e(T_i)\le3\ell_n}}
        \ \sum_{(\tau_1,\ldots,\tau_m)\in
                    \prod_{i=1}^m\mathrm{Tri}(T_i)}
        \prod_{i=1}^m w_n(T_i)
        =
        \left(
            \sum\limits_{\substack{T:\,N_{\Gamma_n}(T)>0\\
                    \ell_n\le\mathsf e(T)\le3\ell_n}}
            \beta_n(T)
        \right)^m.
    \end{equation*}
    The inner sum on the right-hand side is $o(1)$ by \eqref{eq:beta-low-degree}, since $3\ell_n\le D_n$. 
    Summing over $m\ge1$ therefore gives
    \begin{equation}
        \sum\limits_{\substack{F:\,N_{\Gamma_n}(F)>0,\ F\text{ connected}\\\mathsf{e}(F)\ge \ell_n}}
        (1-\varepsilon)^{2\mathsf{e}(F)}\beta_n(F)
        \le
        \sum_{m\ge1}
        \left(
            \sum\limits_{\substack{T:\,N_{\Gamma_n}(T)>0\\\ell_n\le\mathsf{e}(T)\le3\ell_n}}\beta_n(T)
        \right)^m
        =o(1).
        \label{eq:upper-bound-big}
    \end{equation}
    Combining \eqref{eq:upper-bound-small} and
    \eqref{eq:upper-bound-big} proves that
    $\theta_{n,\varepsilon}=o(1)$.
\end{proof}

We can now complete the proof of Theorem \ref{thm:main}.
\begin{proof}[Proof of Theorem \ref{thm:main}]
    Under \eqref{eq:tree-width-condition} and \eqref{eq:low-degree-condition}, Lemma \ref{lem:connect-upper-bound} implies that $\theta_{n,\varepsilon}=o(1)$.
    Consequently, by Lemma \ref{lem:cluster-expansion}, \eqref{eq:beta-noised-target} holds.
    It then follows from Lemma \ref{lem:leading-term} and Lemma \ref{lem:edge} that \begin{equation*}
        \left\lVert L_{n,\varepsilon}\right\rVert_{L^2(Q_n)}^2
    =1+o(1).
    \end{equation*}

    Finally, the Cauchy--Schwarz inequality gives
    \begin{equation*}
            \left\lVert T_{\varepsilon}P_n-Q_n\right\rVert_{\mathrm{TV}}
            =\frac{1}{2}\mathbb{E}_{Q_n} |L_{n,\varepsilon}-1|\le \frac{1}{2}\sqrt{\left\lVert L_{n,\varepsilon}\right\rVert_{L^2(Q_n)}^2-1}=o(1).
    \end{equation*}
\end{proof}
\subsection{The Critical and Subcritical Regimes}
We now turn to $c\in(0,1]$. 
Under \eqref{eq:low-degree-condition}, the graph $T_{\varepsilon}\Gamma_n$ is a forest with high probability, so we need no a priori assumption on the treewidth of $\Gamma_n$. 
We use this to prove \cref{thm:subcritical}.
\begin{lemma}
    Fix $c\in(0,1]$ and $\varepsilon\in(0,1)$. Let
    $\{\Gamma_n\}_{n\ge 1}$ be a sequence of simple graphs satisfying
    $|V(\Gamma_n)|\le n$.
    Suppose that $D_n=\omega(\log n)$ and that
    \eqref{eq:low-degree-condition} holds. 
    Then $T_{\varepsilon}\Gamma_n$ is asymptotically acyclic; more precisely,
    \begin{equation*}
        \mathbb{P}\left(T_{\varepsilon}\Gamma_n \text{ contains a cycle}\right)=o(1).
    \end{equation*}
    \label{lem:no-cycle}
\end{lemma}
\begin{proof}
    It suffices to prove that the expected number of cycles in $T_{\varepsilon}\Gamma_n$ tends to zero. For each integer $k\ge 3$, let $\mathcal{C}_k(\Gamma_n)$ denote the set of simple cycles of length $k$ in $\Gamma_n$.
   
    By Lemma \ref{lem:girth}, for any $3\le k\le D_n$, there is no cycle of length $k$ and hence $|\mathcal{C}_k(T_{\varepsilon}\Gamma_n)|=0$.
    For $k>D_n$, we first derive an upper bound on $|\mathcal{C}_k(\Gamma_n)|$.
    For simplicity, write \begin{equation*}
         r_n=\lfloor \frac{D_n-1}{2}\rfloor.
    \end{equation*}
    We claim that for any two vertices $u,v\in V(\Gamma_n)$, there is at most one simple path from $u$ to $v$ of length at most $r_n$.
    Indeed, suppose there were two distinct simple paths $P_1$ and $P_2$ from $u$ to $v$, each of length at most $r_n$. Their union contains a cycle whose length is at most $2r_n< D_n$.

    Now fix $k\in [D_n+1,n]$ and fix an arbitrary total order on
    $V(\Gamma_n)$. Represent each $k$-cycle by the lexicographically first
    of its $2k$ oriented vertex sequences
    $(v_0,v_1,\ldots,v_{k-1})$, with indices understood modulo $k$, and
    consider the tuple

    \begin{equation*}
        (v_0,v_{r_n},v_{2r_n},\ldots,v_{mr_n}),
    \end{equation*}
    where $m=\lfloor \frac{k}{r_n}\rfloor$.
    Between two consecutive vertices in the tuple, and also from $v_{mr_n}$ back to $v_0$, the corresponding segment of the cycle has length at most $r_n$.
    Each segment is uniquely determined by its two endpoints; otherwise,
    the union of two such paths would contain a cycle of length at most
    $2r_n\le D_n$. Therefore every such tuple corresponds to at most one
    cycle of length $k$. If $r_n$ divides $k$, the final entry in the tuple
    is $v_0$, which causes no ambiguity.
    It follows that \begin{equation*}
        |\mathcal{C}_k(\Gamma_n)|\le n^{\frac{k}{r_n}+1}.
    \end{equation*}
    Since every edge in $\Gamma_n$ is retained in $T_{\varepsilon}\Gamma_n$ with probability $1-\varepsilon$ independently, the expected number of cycles in
    $T_{\varepsilon}\Gamma_n$ is at most
    \begin{align*}
        \sum_{k=D_n+1}^{n}
        n^{k/r_n+1}(1-\varepsilon)^k
        &\le
        n\sum_{k=D_n+1}^{\infty}
        \left[
            (1-\varepsilon)
            \exp\left(\frac{\log n}{r_n}\right)
        \right]^k\\
        &\le
        \frac{2n}{\varepsilon}
        \left(1-\frac{\varepsilon}{2}\right)^{D_n+1}
        =o(1).
    \end{align*}
    Indeed, $r_n=\lfloor(D_n-1)/2\rfloor=\omega(\log n)$, so
    \begin{equation*}
        (1-\varepsilon)
        \exp\left(\frac{\log n}{r_n}\right)
        \le 1-\frac{\varepsilon}{2}
    \end{equation*}
    for all sufficiently large $n$. The final estimate follows from
    $D_n=\omega(\log n)$.
  
\end{proof}

We now prove Theorem \ref{thm:subcritical} by comparing the forest-truncated version of the noised model with the null distribution.

\begin{proof}[Proof of Theorem \ref{thm:subcritical}]
    As in the noised model, each edge of $\Gamma_n$ is retained independently with probability $1-\varepsilon$. 
    If the resulting planted subgraph contains a cycle, we replace it by the edgeless graph.
    Let $T_{\varepsilon}^{\mathrm{F}}\Gamma_n$ and
    $T_{\varepsilon}^{\mathrm{F}}P_n$ denote the resulting random planted subgraph and the corresponding planted distribution, respectively.
    Under the natural coupling,
    \begin{equation*}
        T_{\varepsilon}^{\mathrm{F}}\Gamma_n\subset T_{\varepsilon}\Gamma_n\qquad\text{almost surely}.
    \end{equation*}
    Thus, $T_{\varepsilon}^{\mathrm{F}}\Gamma_n$ is stochastically dominated by $T_{\varepsilon}\Gamma_n$ with respect to edge inclusion.
    Moreover, the two corresponding planted distributions agree unless $T_{\varepsilon}\Gamma_n$ contains a cycle. 
    Lemma \ref{lem:no-cycle} therefore gives
    \begin{equation*}
        \left\lVert T_{\varepsilon}^{\mathrm{F}}P_n-T_{\varepsilon}P_n \right\rVert_{\mathrm{TV}}=o(1).
    \end{equation*}
    By the triangle inequality, it remains to prove that
    \begin{equation}
        \left\lVert T_{\varepsilon}^{\mathrm{F}}P_n-Q_n\right\rVert_{\mathrm{TV}}=o(1).
        \label{eq:truncated-tv-goal}
    \end{equation}

    Let
    $L_{n,\varepsilon}^{\mathrm{F}}
    =\frac{dT_{\varepsilon}^{\mathrm{F}}P_n}{dQ_n}$ denote the likelihood ratio of $T_{\varepsilon}^{\mathrm{F}}P_n$ with respect to $Q_n$.
    For every abstract graph $F$ and every embedded copy
    $F^{\bullet}\simeq F$, the Fourier coefficient formula gives
    \begin{equation*}
        \widehat{L_{n,\varepsilon}^{\mathrm F}}(F^\bullet)
        =\left(\frac{1-c/n}{c/n}\right)^{\mathsf e(F)/2}
        \mathbb P\!\left(F^\bullet\subset
        \phi_n(T_\varepsilon^{\mathrm F}\Gamma_n)\right).
    \end{equation*}
    Under the coupling
    $T_\varepsilon^{\mathrm F}\Gamma_n\subseteq T_\varepsilon\Gamma_n$,
    the containment probability is at most its counterpart in the original
    noised model. Hence
    \begin{equation*}
        \widehat{L_{n,\varepsilon}^{\mathrm{F}}}(F^{\bullet})\le
        \widehat{L_{n,\varepsilon}}(F^{\bullet}).
    \end{equation*}
    If $F$ contains a cycle, then
    \begin{equation*}
        \widehat{L_{n,\varepsilon}^{\mathrm{F}}}(F^{\bullet})=0,
    \end{equation*}
    because $T_{\varepsilon}^{\mathrm{F}}\Gamma_n$ is always a forest.
    Applying Parseval's identity as in the derivation of
    \eqref{eq:noise-lr}, we obtain
    \begin{equation*}
        \left\lVert L_{n,\varepsilon}^{\mathrm{F}}\right\rVert_{L^2(Q_n)}^2\le
        1+\sum\limits_{\substack{F:\,N_{\Gamma_n}(F)>0,\ F\text{ is a forest}\\1\le\mathsf{e}(F)}}(1-\varepsilon)^{2\mathsf{e}(F)}\left(\frac{1-c/n}{c/n}\right)^{\mathsf{e}(F)}\frac{N_{\Gamma_n}(F)^2|\mathrm{Aut}(F)|}{(n)_{\mathsf{v}(F)}}.
    \end{equation*}
    As in the proof of Theorem \ref{thm:main}, Lemma \ref{lem:leading-term} and Lemma \ref{lem:edge} reduce the desired bound \eqref{eq:truncated-tv-goal} to
    \begin{equation*}
        \sum\limits_{\substack{F:\,N_{\Gamma_n}(F)>0,\ F\text{ is a forest}\\1\le\mathsf{e}(F)}}(1-\varepsilon)^{2\mathsf{e}(F)}\beta_n(F)=o(1).
    \end{equation*}
    Applying the argument of Lemma \ref{lem:cluster-expansion} within the
    collection of forests, it suffices to show that
    \begin{equation}
        \sum\limits_{\substack{F:\,N_{\Gamma_n}(F)>0,\ F\text{ is a tree}\\1\le\mathsf{e}(F)}}(1-\varepsilon)^{2\mathsf{e}(F)}\beta_n(F)=o(1).
        \label{eq:truncated-sum-goal}
    \end{equation}

    Since $D_n=\omega(\log n)$, choose positive integers $\ell_n$ such
    that
    \[
        3\ell_n\le D_n,
        \qquad
        2\log n=o(\ell_n).
    \]
    By \cref{prop:treewidth-division} with $k=1$, every tree $F$ with $N_{\Gamma_n}(F)>0$ and $\mathsf{e}(F)\ge\ell_n$ admits an ordered edge decomposition with $\ell_n\le\mathsf{e}(T_i)\le 3\ell_n$ and $\mathrm{ov}(F,\mathcal{T}(F))\le 2(m-1)$.
    The proof of \cref{lem:connect-upper-bound} uses \eqref{eq:tree-width-condition} only through such decompositions and the choice of $\ell_n$, so the same argument, with $\operatorname{tw}(\Gamma_n)+1$ replaced by $2$ and the sum restricted to trees, shows that the sum in \eqref{eq:truncated-sum-goal} is $o(1)$.
    It follows that
    \begin{equation*}
        \left\lVert L_{n,\varepsilon}^{\mathrm{F}}\right\rVert_{L^2(Q_n)}^2=1+o(1).
    \end{equation*}
    The Cauchy--Schwarz inequality now yields
    \begin{equation*}
        \left\lVert T_{\varepsilon}^{\mathrm{F}}P_n-Q_n\right\rVert_{\mathrm{TV}}\le\frac{1}{2}
        \sqrt{\left\lVert L_{n,\varepsilon}^{\mathrm{F}}\right\rVert_{L^2(Q_n)}^2-1}=o(1).
    \end{equation*}
    This proves \eqref{eq:truncated-tv-goal}, and the preceding coupling completes the proof.
\end{proof}
\section{Low-Interface Decompositions of Bounded-Treewidth Graphs}
\label{section:decomposition}

In this section, we prove \cref{lem:decomposition}. Treewidth does not increase when passing to subgraphs, so if $N_{\Gamma}(F)>0$ and $\mathrm{tw}(\Gamma)\le k$, then $\mathrm{tw}(F)\le k$. 
It therefore
suffices to prove the following proposition.

\begin{proposition}
    Let $k,\ell$ be positive integers, and let $F$ be a connected graph satisfying
    \begin{equation*}
        \mathrm{tw}(F)\leq k\qquad\text{and}\qquad \mathsf{e}(F)\ge \ell.
    \end{equation*}
    Then $F$ admits an ordered edge decomposition
    \begin{equation*}
        \mathcal{T}(F)=(T_1,T_2,\ldots,T_m)
    \end{equation*}
    such that
    \begin{equation*}
        \ell\le\mathsf{e}(T_i)\le 3\ell\quad\text{for every }i=1,2,\ldots,m,
        \qquad\mathrm{ov}(F,\mathcal{T}(F))\le(k+1)(m-1).
    \end{equation*}
    \label{prop:treewidth-division}
\end{proposition}

We prove the proposition in two steps. First, in Subsection
\ref{sec:subcubic}, we refine a tree decomposition of $F$ so that its decomposition tree is subcubic and each node owns at most one edge of $F$. 
Second, in Subsection \ref{sec:partition-tree}, we partition the resulting weighted tree into clusters and use these clusters to construct the desired ordered edge decomposition of $F$.

\subsection{Subcubic Edge-Owned Tree Decompositions}
\label{sec:subcubic}
We first recall the definitions of tree decompositions and treewidth.
\begin{definition}
    Let $F$ be a connected graph. A pair $(\mathsf{T},(B_x)_{x\in V(\mathsf T)})$, where $\mathsf T$ is a tree and $B_x\subseteq V(F)$ is a \textit{bag} for each $x\in V(\mathsf T)$, is called a \textit{tree decomposition} of $F$ if the following conditions hold:
    \begin{itemize}
        \item every vertex of $F$ belongs to at least one bag;
        \item for every edge $uv\in E(F)$, there exists a bag containing both $u$ and $v$;
        \item for every vertex $v\in V(F)$, the set
        \begin{equation*}
            \mathsf{T}[v]=\{x\in V(\mathsf{T}):v\in B_x\}
        \end{equation*}
        induces a connected subtree of $\mathsf{T}$.
    \end{itemize}
    The width of the decomposition $(\mathsf{T},(B_x)_{x\in V(\mathsf T)})$ is $\max_x |B_x|-1$. 
    The \textit{treewidth} $\mathrm{tw}(F)$ of $F$ is the minimum width among all tree decompositions of $F$.
\end{definition}

To construct an edge decomposition, we assign each edge of $F$ to a node of $\mathsf{T}$. 
To control the number of edges in each piece,
we want every node of $\mathsf{T}$ to be assigned at most one edge and $\mathsf{T}$ to have bounded degree.

\begin{lemma}
    \label{lem:decomposition-refinement}
    For a connected graph $F$, if $\mathrm{tw}(F)\le k$, then $F$ has a tree decomposition  $(\mathsf{T},(B_x)_{x\in V(\mathsf T)})$ of width at most $k$ together with a map,
    \begin{equation*}
        \pi: E(F)\longrightarrow V(\mathsf T),
    \end{equation*}
    such that  \begin{itemize}
        \item the endpoints of every edge $e$ lie in $B_{\pi(e)}$;
        \item each node of $\mathsf T$ owns at most one edge, in the sense that $|\pi^{-1}(x)|\le 1$;
        \item $\mathsf T$ has maximum degree at most $3$.
    \end{itemize}
\end{lemma}
\begin{proof}

    Start with a tree decomposition $(\mathsf T_0,(B_x^0))$ of width at most $k$. 
    For every edge $e\in E(F)$, we choose a host node $x(e)\in V(\mathsf T_0)$ such that the bag $B_{x(e)}^0$ contains both endpoints of $e$.

    For each $x\in V(\mathsf T_0)$, define
    \begin{equation*}
         d_x:=\deg_{\mathsf T_0}(x),\qquad r_x:=|\{e\in E(F):x(e)=x\}|.
    \end{equation*}
    Replace $x$ with a finite tree $R_x$ of maximum degree at most $3$ that contains at least $d_x+r_x$ designated vertices of degree at most $1$, and assign the bag $B_x^0$ to every vertex of $R_x$. 
    Of these designated vertices, choose $d_x$ as distinct ports, one for each edge of $\mathsf T_0$ incident to $x$, and choose $r_x$ additional vertices, one for each edge $e\in E(F)$ satisfying $x(e)=x$. 
    Declare each such vertex to be the owner of the corresponding edge.

    We then assemble the family $(R_x)_{x\in V(\mathsf T_0)}$ according to the adjacency structure of $\mathsf T_0$. 
    More precisely, for each edge $xy\in E(\mathsf T_0)$, join the port of $R_x$ corresponding to $xy$ to the port of $R_y$ corresponding to $xy$. Since $\mathsf T_0$ and each $R_x$ are trees, the resulting graph, denoted by $\mathsf T$, is also a tree.

    Thus, we have obtained the pair $(\mathsf{T},(B_x)_{x\in V(\mathsf{T})})$. We first verify that it is indeed a tree decomposition.
    Every vertex of $F$ belongs to some bag $B_x^0$ of $\mathsf T_0$, and hence to a corresponding bag $B_y$ with $y\in V(R_x)$.
    Similarly, every edge of $F$ has both endpoints in some bag.
    Finally, for every vertex $u\in V(F)$, since $\mathsf T_0[u]$ induces a connected subtree of $\mathsf T_0$, the subgraph of $\mathsf T$ induced by $\bigcup_{x\in V(\mathsf T_0[u])}V(R_x)$ is connected.

    We next verify that the tree decomposition $(\mathsf T,(B_x)_{x\in V(\mathsf T)})$ constructed above satisfies the required properties.
    Since every bag indexed by a vertex of $\mathsf T$ is a copy of an original bag in $(B_y^0)_{y\in V(\mathsf T_0)}$, the width of the tree decomposition $(\mathsf{T},(B_x)_{x\in V(\mathsf T)})$ is at most $k$.
    Each edge in $E(F)$ has an owner in $V(\mathsf{T})$, thereby yielding a map $\pi:E(F)\to V(\mathsf T)$ such that each vertex owns at most one edge.
    For every edge $e\in E(F)$, both endpoints of $e$ belong to the bag $B_{x(e)}^0$, which coincides with $B_{\pi(e)}$.
    Moreover, adding one external edge at a port vertex preserves maximum degree $3$ in $\mathsf T$.

\end{proof}
\subsection{Partitioning a Subcubic Tree}
\label{sec:partition-tree}
By Lemma \ref{lem:decomposition-refinement}, every connected graph $F$ of treewidth at most $k$ admits a tree decomposition of width at most $k$ whose decomposition tree is subcubic and in which each edge of $F$ is assigned to a distinct vertex. 
We now cut $F$ into subgraphs along edges of the decomposition tree; the width of the decomposition then bounds the vertex-overlap count of the resulting ordered edge decomposition.

The next lemma chooses the edges of $\mathsf{T}$ along which we cut, so that all pieces have comparable numbers of edges.

\begin{lemma}
    Let $\mathsf T$ be a finite tree of maximum degree at most $3$ whose vertices have weights in $\{0,1\}$, and let its total weight be $t$.
    For every integer $\ell\ge 1$ with $t\ge \ell$, there is a partition \begin{equation*}
        V(\mathsf T)=\mathcal C_1\sqcup \mathcal C_2\sqcup \cdots \sqcup \mathcal C_m,
    \end{equation*}
    such that each $\mathcal{C}_i$ induces a connected subtree of $\mathsf{T}$ and has total weight in $[\ell,3\ell]$.
    Moreover, after contracting each $\mathcal C_i$, the quotient graph is a tree.
    \label{lem:partition-tree}
\end{lemma}

\begin{proof}
    Root $\mathsf{T}$ at a leaf, so every vertex has at most two children.
    We iteratively cut off rooted subtrees, proceeding from the leaves toward the root.
    When the current remaining tree has weight at least $\ell$, we choose a deepest vertex $x$ whose current descendant subtree $D_x$ has weight at least $\ell$.
    Then every child subtree of $x$ has weight at most $\ell-1$.
    Since $x$ has at most two children and its own weight is at most $1$, the total weight of $D_x$ is at most $1+2(\ell-1)=2\ell-1$.
    Record $D_x$ as a cluster and remove it from the tree. The remaining part is connected.

    When this procedure stops, the remaining tree $R$ has weight at most $\ell-1$.
    Since the original total weight is at least $\ell$, there is at least one cluster recorded.
    If $R$ is nonempty, then it contains the root, so it is connected and adjacent to some recorded cluster; we merge it into one such cluster.
    The merged cluster is connected and has weight at most $3\ell-2$.
    Therefore, every cluster has weight between $\ell$ and $3\ell$.
    The clusters partition the vertices of a tree into connected subtrees, so contracting them preserves connectedness and cannot create a cycle. Hence, the quotient graph is a tree.
\end{proof}

We conclude this section by proving Proposition \ref{prop:treewidth-division}.
\begin{proof}[Proof of Proposition \ref{prop:treewidth-division}]
    Applying Lemma \ref{lem:decomposition-refinement} to $F$ yields a subcubic tree decomposition $(\mathsf{T},(B_x))$ and an ownership map $\pi:E(F)\to V(\mathsf T)$ such that every vertex of $\mathsf T$ owns at most one edge.
    Give a vertex weight $1$ if it owns an edge and weight $0$ otherwise.
    Then the total weight is $|E(F)|$.

    Lemma \ref{lem:partition-tree} divides $V(\mathsf{T})$ into connected clusters \begin{equation*}
        V(\mathsf T)=\mathcal C_1\sqcup \mathcal C_2\sqcup \cdots \sqcup \mathcal C_m,
    \end{equation*} with weights between $\ell$ and $3\ell$.
    Let $T_i$ denote the graph whose edge set consists of the edges owned by the vertices in $\mathcal C_i$ and whose vertex set consists of the endpoints of those edges.
    Every edge is owned by a unique vertex, and hence the $T_i$ form an ordered edge decomposition of $F$.
    Moreover, $\mathsf e(T_i)$ equals the total weight of $\mathcal{C}_i$ and therefore belongs to $[\ell,3\ell]$.

    It suffices to show that $ \mathrm{ov}(F,\mathcal{T}(F))\le(k+1)(m-1)$.
    For every vertex $v\in V(F)$, write $I_v:=\{i:v\in V(T_i)\}$.
    Since $F$ is connected and has at least one edge, it has no isolated vertices, so $I_v$ is nonempty.
    Hence, \begin{equation}
        \mathrm{ov}(F,\mathcal{T}(F))=\sum\limits_{i=1}^m \mathsf{v}(T_i)-\mathsf{v}(F)=\sum\limits_{v\in V(F)}(|I_v|-1).
        \label{eq:overlap-decomposition}
    \end{equation}
    Let $\mathsf Q$ denote the quotient tree obtained by contracting all clusters $\mathcal C_i$.
    The vertices in $\mathsf T$ whose bags contain $v$ form a connected subtree $\mathsf T[v]$. Let $\mathsf{Q}[v]$ denote its image under contraction.
    Then $\mathsf{Q}[v]$ is a connected subtree of $\mathsf{Q}$.

    If $v\in V(T_i)$, then an edge incident to $v$ is owned by a vertex in $\mathcal C_i$, and the owner's bag contains $v$.
    Hence $I_v\subset V(\mathsf Q[v])$, and consequently
    \begin{equation}
        |I_v|-1\le |E(\mathsf Q[v])|.
        \label{eq:overlap-decomposition-edge}
    \end{equation}
    
    For each edge $ij\in E(\mathsf{Q})$, since $\mathsf{T}$ is a tree, there is a unique edge $x_{ij}y_{ij}\in E(\mathsf T)$ with $x_{ij}\in\mathcal C_i$ and $y_{ij}\in\mathcal C_j$.
    If the edge $ij$ belongs to some $\mathsf{Q}[v]$, then both $x_{ij}$ and $y_{ij}$ belong to $\mathsf{T}[v]$, which implies $v\in B_{x_{ij}}\cap B_{y_{ij}}$.
    Since the width of the tree decomposition is at most $k$, $|B_{x_{ij}}\cap B_{y_{ij}}|\le k+1$.
    Consequently,
    \begin{equation*}
        \sum\limits_{v\in V(F)} (|I_v|-1)\le \sum\limits_{v\in V(F)} |E(\mathsf Q[v])|\le \sum\limits_{ij\in E(\mathsf{Q})} |B_{x_{ij}}\cap B_{y_{ij}}|\le(k+1)(m-1),
    \end{equation*}
    where the final inequality uses that $\mathsf Q$ is a tree.
    Combining this with \eqref{eq:overlap-decomposition} completes the proof.
\end{proof}

\appendix

\section{A Low-Degree--Statistical Gap in the Supercritical Regime}
\label{sec:super-quenched-gap}
In this appendix, we construct a planted-subgraph model in the supercritical regime with a gap between low-degree and statistical distinguishability.
We first list the properties we need from a deterministic sequence of planted graphs, and show that for every sufficiently large fixed $c$ the resulting model satisfies the low-degree condition \eqref{eq:low-degree-condition} but remains distinguishable from the null even after noise. 
In the last subsection, we construct such a sequence by starting from random $D$-regular graphs and removing their short cycles.

Fix constants
\begin{equation*}
    r\in(0,1/2),\qquad \gamma\in(0,1),\qquad\varepsilon\in(0,1),
\end{equation*}
and choose an even integer $D \ge 6$ such that
\begin{equation}
    \frac{(1-\varepsilon)D }{2}>1.
    \label{eq:deterministic-degree-choice}
\end{equation}
Set
\begin{equation*}
    m_n:=\lfloor n^r\rfloor,\qquad
    D_n:=\left\lceil(\log n)^2\right\rceil.
\end{equation*}

\begin{lemma}
    There exist constants $A=A(D ,r,\gamma)>0$ and $\kappa=\kappa(D ,r)>0$, together with a deterministic sequence of simple graphs $\{\Gamma_n\}_{n\ge1}$ on $[m_n]$, satisfying the following conditions,
    \begin{align}
        \max_{v\in[m_n]}\deg_{\Gamma_n}(v)&\le D ,
        \label{eq:deterministic-max-degree}\\
        \mathsf e(\Gamma_n[U])
        &\le \left(1+\frac{A}{\log n}\right)|U|
        \quad\text{for every }U\subseteq[m_n]
        \text{ with }2\le |U|\le m_n^\gamma,
        \label{eq:deterministic-local-sparsity}\\
        \operatorname{girth}(\Gamma_n)&\ge \kappa\log n,
        \label{eq:deterministic-girth}\\
        \frac{\mathsf e(\Gamma_n)}{m_n}&=\frac{D }{2}-o(1).
        \label{eq:deterministic-edge-density}
    \end{align}
    \label{lem:deterministic-planted-graphs}
\end{lemma}
\begin{remark}
    \label{rmk:deterministic}
    The first three conditions control $\beta_n(F)$ uniformly over all graphs $F$ in the low-degree expansion. 
    Bounded degree \eqref{eq:deterministic-max-degree} controls the isomorphism-triple factor $N_{\Gamma_n}(F)^2|\operatorname{Aut}(F)|$. 
    For cyclic $F$, local sparsity \eqref{eq:deterministic-local-sparsity} bounds the excess $\mathsf e(F)-\mathsf v(F)$ and hence the factor $n^{\mathsf e(F)-\mathsf v(F)}$, while logarithmic girth \eqref{eq:deterministic-girth} forces $F$ to have many edges. 
    The positive edge density \eqref{eq:deterministic-edge-density} ensures that enough planted edges survive the noise for a scan test to detect them. 
\end{remark}

Having fixed the sequence from Lemma \ref{lem:deterministic-planted-graphs}, let $Q_n=G(n,c/n)$ and define $P_n^{\mathrm{reg}}$ by planting a uniformly random copy of $\Gamma_n$ in an independent sample from $Q_n$, as in Section \ref{section:model}.
Write
\begin{equation*}
    L_n^{\mathrm{reg}}:=\frac{dP_n^{\mathrm{reg}}}{dQ_n},\qquad(L_n^{\mathrm{reg}})^{\le D_n}:=\Pi_{\le D_n}^{Q_n}(L_n^{\mathrm{reg}}).
\end{equation*}

\begin{theorem}
    There exists a constant $c_0=c_0(D ,r,\gamma)>1$ such that, for every fixed $c\ge c_0$,
    \begin{equation}
        \left\lVert(L_n^{\mathrm{reg}})^{\le D_n}
        \right\rVert_{L^2(Q_n)}^2=1+o(1).
        \label{eq:deterministic-low-degree}
    \end{equation}
    Nevertheless,
    \begin{equation}
        \left\lVert T_\varepsilon P_n^{\mathrm{reg}}-Q_n
        \right\rVert_{\mathrm{TV}}=1-o(1).
        \label{eq:deterministic-statistical}
    \end{equation}
    Thus, the low-degree--statistical gap holds for a deterministic
    sequence of planted graphs.
    \label{thm:deterministic-gap}
\end{theorem}

\subsection{Low-Degree Indistinguishability}
\label{subsec:deterministic-low-degree}

We now prove \eqref{eq:deterministic-low-degree}, following the outline in \cref{rmk:deterministic}.
We begin with a counting estimate for bounded-degree graphs.
For an abstract graph $F$, let
\begin{equation*}
    \operatorname{Emb}_{\Gamma_n}(F) :=N_{\Gamma_n}(F)|\operatorname{Aut}(F)|
\end{equation*}
denote the number of injective embeddings of $F$ into $\Gamma_n$.

\begin{lemma}
    Let $\Gamma$ be a graph on $m$ vertices with maximum degree at most $D $.
    For every $e\ge1$,
    \begin{equation}
        \sum_{\substack{F:\ F\text{ connected}\\
        \mathsf e(F)=e}}N_\Gamma(F)^2|\operatorname{Aut}(F)|\le m^2D ^{3e}. \label{eq:deterministic-connected-count}
    \end{equation}
    \label{lem:deterministic-connected-count}
\end{lemma}

\begin{proof}
    Fix a connected graph $F$ with $v$ vertices and $e$ edges, together with a root and a spanning tree.
    After choosing the image of the root, the spanning tree can be embedded by choosing the image of each new vertex among at most $D $ neighbors of the image of its parent.
    Ignoring injectivity and the remaining edge constraints only enlarges the count, and hence
    \begin{equation*}
        \operatorname{Emb}_{\Gamma}(F) \le mD ^{v-1}\le mD ^e.
    \end{equation*}
    Moreover, the number of connected $e$-edge subgraphs of $\Gamma$ is at most $mD ^{2e}$.
    Indeed, fix orders on the vertices and incident edges of $\Gamma$.
    For each connected subgraph, root its doubled-edge multigraph at its
    least vertex and choose the lexicographically first Euler tour. The
    initial vertex and the resulting sequence of $2e$ local edge choices
    determine the subgraph, giving at most $mD ^{2e}$ possibilities.
    Therefore,
    \begin{align*}
        \sum_{\substack{F:\ F\text{ connected}\\
        \mathsf e(F)=e}} N_\Gamma(F)^2|\operatorname{Aut}(F)| &=\sum_F N_\Gamma(F)\operatorname{Emb}_\Gamma(F)\\
        &\le mD ^e\sum_F N_\Gamma(F) \le m^2D ^{3e},
    \end{align*}
    as claimed.
\end{proof}

We next combine Lemma \ref{lem:deterministic-connected-count} with the Fourier expansion \eqref{eq:ldlr-2} to prove \eqref{eq:deterministic-low-degree}.

\begin{proof}[Proof of \eqref{eq:deterministic-low-degree} in Theorem \ref{thm:deterministic-gap}]
    Recall from \eqref{eq:ldlr-2} and Lemma \ref{lem:leading-term} that, uniformly over $1\le\mathsf e(F)\le D_n$, the contribution of $F$ to the squared low-degree norm is at most twice
    \begin{equation*}
        \beta_n(F) =\frac{n^{\mathsf e(F)-\mathsf v(F)}} {c^{\mathsf e(F)}} N_{\Gamma_n}(F)^2|\operatorname{Aut}(F)|.
    \end{equation*}
    We begin by summing this quantity over connected graphs.

    If $F$ is a tree with $e$ edges, then $\mathsf v(F)=e+1$.
    Since $m_n^2/n=o(1)$, Lemma \ref{lem:deterministic-connected-count} gives
    \begin{equation}
        \sum_{\substack{F:\ F\text{ a tree}\\
        1\le\mathsf e(F)\le D_n}} \beta_n(F) \le \frac{m_n^2}{n} \sum_{e=1}^{D_n}\left(\frac{D^3 }{c}\right)^e =o(1) \label{eq:deterministic-tree-sum}
    \end{equation}
    whenever $c>D^3 $.

    Suppose next that $F$ is connected and contains a cycle.
    Write $v=\mathsf v(F)$ and $e=\mathsf e(F)$.
    Since $F$ is connected and contains a cycle, $v\le e\le D_n\le m_n^\gamma$ for all sufficiently large $n$.
    The local-sparsity condition \eqref{eq:deterministic-local-sparsity} therefore implies
    \begin{equation*}
        e-v\le\frac{A}{\log n}v \le\frac{A}{\log n}e,
    \end{equation*}
    while \eqref{eq:deterministic-girth} implies $e\ge\kappa\log n$.
    Consequently, for all sufficiently large $n$,
    \begin{equation*}
        n^{e-v}\le\exp(Ae).
    \end{equation*}
    Applying Lemma \ref{lem:deterministic-connected-count} for each $e$ yields
    \begin{equation}
        \sum_{\substack{F:\ F\text{ connected and cyclic}\\
        1\le\mathsf e(F)\le D_n}} \beta_n(F) \le m_n^2 \sum_{e\ge\kappa\log n} \left(\frac{D^3 \exp(A)}{c}\right)^e. \label{eq:deterministic-cyclic-sum}
    \end{equation}
    We may choose $c_0$ sufficiently large that, for $c\ge c_0$,
    \begin{equation*}
        \theta:=\frac{D^3 \exp(A)}{c} <\exp\left(-\frac{3r}{\kappa}\right).
    \end{equation*}
    The right-hand side of \eqref{eq:deterministic-cyclic-sum} is then
    $o(1)$, because
    \begin{equation*}
        m_n^2\sum_{e\ge\lceil\kappa\log n\rceil}\theta^e
        =\frac{m_n^2\theta^{\lceil\kappa\log n\rceil}}{1-\theta}
        \le\frac{n^{2r}\exp(-3r\log n)}{1-\theta}
        =o(1).
    \end{equation*}
    Together with \eqref{eq:deterministic-tree-sum}, this proves
    \begin{equation*}
        \sum_{\substack{F:\ F\text{ connected}\\
        1\le\mathsf e(F)\le D_n}}\beta_n(F)=o(1).
    \end{equation*}

    It remains to pass from connected graphs to arbitrary graphs.
    Write an abstract graph $F$ as a disjoint union of pairwise nonisomorphic connected types $H$, with multiplicities $k_H$.
    Then
    \begin{equation*}
        |\operatorname{Aut}(F)| =\prod_H |\operatorname{Aut}(H)|^{k_H}k_H!,\qquad \operatorname{Emb}_{\Gamma_n}(F) \le\prod_H\operatorname{Emb}_{\Gamma_n}(H)^{k_H}.
    \end{equation*}
    Since $N_{\Gamma_n}(F)^2|\operatorname{Aut}(F)| =\operatorname{Emb}_{\Gamma_n}(F)^2/ |\operatorname{Aut}(F)|$, it follows that
    \begin{equation*}
        \beta_n(F) \le\prod_H\frac{\beta_n(H)^{k_H}}{k_H!}.
    \end{equation*}
    Summing over all multisets of connected components gives
    \begin{equation*}
        \sum_{\substack{F:\ 1\le\mathsf e(F)\le D_n}} \beta_n(F) \le\exp\left( \sum_{\substack{F:\ F\text{ connected}\\
        1\le\mathsf e(F)\le D_n}}\beta_n(F)\right)-1=o(1).
    \end{equation*}
    The exact Fourier formula \eqref{eq:ldlr-2} and Lemma \ref{lem:leading-term} now prove \eqref{eq:deterministic-low-degree}.
\end{proof}

\subsection{Statistical Distinguishability}
\label{subsec:deterministic-statistical}

We prove \eqref{eq:deterministic-statistical} with a scan test, which thresholds the maximum number of edges induced by a set of $m_n$ vertices.

\begin{proof}[Proof of \eqref{eq:deterministic-statistical} in Theorem \ref{thm:deterministic-gap}]
    Choose a constant $\tau$ such that
    \begin{equation*}
        1<\tau<\frac{(1-\varepsilon)D }{2},
    \end{equation*}
    which is possible by \eqref{eq:deterministic-degree-choice}.
    For a graph $G$ on $[n]$, define
    \begin{equation*}
        M_n(G):=\max_{\substack{S\subset[n]\\
        |S|=m_n}} \mathsf e(G[S]),
    \end{equation*}
    and put $t_n:=\lceil\tau m_n\rceil$.

    Under $Q_n$, a union bound over the choices of $S$ and of $t_n$ edges within $S$ gives
    \begin{align*}
        \mathbb P_{Q_n}(M_n\ge t_n) &\le\binom{n}{m_n} \binom{\binom{m_n}{2}}{t_n} \left(\frac{c}{n}\right)^{t_n}\\
        &\le\exp\left\{-(\tau-1)(1-r)m_n\log n +o(m_n\log n)\right\}=o(1).
    \end{align*}

    Under $T_\varepsilon P_n^{\mathrm{reg}}$, each edge of the fixed planted graph $\Gamma_n$ is retained independently with probability $1-\varepsilon$.
    If $R_n$ denotes the number of retained planted edges, then \eqref{eq:deterministic-edge-density} and binomial concentration give
    \begin{equation*}
        \frac{R_n}{m_n}\xrightarrow{\mathbb P} \frac{(1-\varepsilon)D }{2}>\tau.
    \end{equation*}
    Taking $S$ to be the planted vertex set in the definition of $M_n$, we obtain $M_n\ge R_n$.
    Hence
    \begin{equation*}
        \mathbb P_{T_\varepsilon P_n^{\mathrm{reg}}}(M_n\ge t_n)=1-o(1).
    \end{equation*}
    The test that rejects the null when $M_n\ge t_n$ has vanishing total error, which proves \eqref{eq:deterministic-statistical} and completes the proof of Theorem \ref{thm:deterministic-gap}.
\end{proof}

\subsection{Construction of the Deterministic Planted Graphs}
\label{subsec:deterministic-construction}

We prove \cref{lem:deterministic-planted-graphs} using the $D$-regular configuration model.
We first show that, with high probability, the configuration model is locally sparse and has few short cycles.
We then choose a simple realization with both properties and delete one edge from each short cycle.
The resulting graph has logarithmic girth and keeps bounded degree, local sparsity, and asymptotically the full edge density of the regular graph.

For the remainder of this subsection, write $m=m_n$, and let $G_n^{\mathrm{conf}}$ be the $D$-regular configuration model on $[m]$.
In this model, every vertex is equipped with $D$ half-edges, and a uniformly random perfect matching of the $Dm$ half-edges produces a $D$-regular multigraph.
Throughout this subsection, parallel edges are counted with multiplicity, while each loop is counted once as an edge and twice toward the degree of its incident vertex.

We begin with the local sparsity property of the configuration model.
For a fixed vertex set $U$, any collection of $k$ edges in $G_n^{\mathrm{conf}}[U]$ corresponds to $k$ disjoint pairs among the half-edges incident to $U$.
We therefore bound the probability of a dense induced subgraph by first taking a union bound over such collections of pairs and then another union bound over the choices of $U$.
\begin{lemma}
    \label{lem:deterministic-configuration-local-sparsity}
    There exists a constant $A=A(D,r,\gamma)>0$ such that
    \begin{equation}
        \mathbb P\left( \mathsf e(G_n^{\mathrm{conf}}[U]) \le \left(1+\frac{A}{\log n}\right)|U| \text{ for every }U\subseteq[m]\text{ with } 2\le |U|\le m^\gamma \right)=1-o(1). \label{eq:deterministic-configuration-local-event}
    \end{equation}
\end{lemma}

\begin{proof}
    Put
    \begin{equation*}
        \delta_n:=\frac{A}{\log n},\qquad k_s:=\lfloor(1+\delta_n)s\rfloor+1=s+t_s, \qquad t_s:=\lfloor\delta_n s\rfloor+1.
    \end{equation*}
    Fix $U\subseteq[m]$ with $|U|=s\le m^\gamma$.
    If $G_n^{\mathrm{conf}}[U]$ contains at least $k_s$ edges, then at least $k_s$ disjoint pairs among the $Ds$ half-edges incident to $U$ occur in the random perfect matching.
    A union bound over these pairs gives
    \begin{equation*}
        \mathbb P\bigl(\mathsf e(G_n^{\mathrm{conf}}[U])\ge k_s\bigr) \le \frac{(Ds)_{2k_s}}{2^{k_s}k_s!} \prod_{j=0}^{k_s-1}\frac{1}{Dm-(2j+1)}
    \end{equation*}
    whenever the event on the left is possible.
    Uniformly over the stated range, $k_s\le 2s=o(m)$ for all sufficiently large $n$, so the right-hand side is at most
    \begin{equation*}
        \frac{(Ds)^{2k_s}}{2^{k_s}k_s!} \left(\frac{2}{Dm}\right)^{k_s} \le \left(\frac{eDs}{m}\right)^{k_s},
    \end{equation*}
    where the last inequality uses $k_s\ge s$.
    Taking another union bound over the choices of $U$, using $\binom{m}{s}\le(em/s)^s$, and recalling that $k_s=s+t_s$, we obtain
    \begin{equation*}
        \begin{aligned}
            &\mathbb P\left(\exists U\subseteq[m]: |U|=s, \ \mathsf e(G_n^{\mathrm{conf}}[U])\ge k_s\right)\\
            &\qquad\le \left(\frac{em}{s}\right)^s \left(\frac{eDs}{m}\right)^{s+t_s} =(e^2D)^s(eD)^{t_s}\left(\frac{s}{m}\right)^{t_s}.
        \end{aligned}
    \end{equation*}
    Moreover, since $k_s\le2s$, we have $t_s\le s$.
    Hence, setting $C_D:=e^3D^2>1$, the preceding display yields
    \begin{equation}
        \mathbb P\left(\exists U\subseteq[m]: |U|=s, \ \mathsf e(G_n^{\mathrm{conf}}[U])\ge k_s\right) \le C_D^s\left(\frac{s}{m}\right)^{t_s}. \label{eq:deterministic-configuration-local-bound}
    \end{equation}

    Choose $a>0$ sufficiently small that $a\log C_D<1/2$.
    For $2\le s\le a\log m$, we have $t_s\ge1$, and therefore
    \begin{equation*}
        \begin{aligned}
            \sum_{2\le s\le a\log m} C_D^s\left(\frac{s}{m}\right)^{t_s} &\le \frac{1}{m}\sum_{2\le s\le a\log m}sC_D^s\\
            &=O\left((\log m)m^{-1+a\log C_D}\right) =o(1).
        \end{aligned}
    \end{equation*}
    For $a\log m<s\le m^\gamma$, we have $t_s>As/\log n$ and
    \begin{equation*}
        \log\frac{m}{s} \ge (1-\gamma)\log m =\bigl(r(1-\gamma)+o(1)\bigr)\log n.
    \end{equation*}
    Hence
    \begin{equation*}
        \begin{aligned}
            C_D^s\left(\frac{s}{m}\right)^{t_s} &=\exp\left\{s\log C_D-t_s\log\frac{m}{s}\right\}\\
            &\le \exp\left\{ \bigl(\log C_D-Ar(1-\gamma)+o(1)\bigr)s \right\}.
        \end{aligned}
    \end{equation*}
    Choose $A=A(D,r,\gamma)$ sufficiently large that $Ar(1-\gamma)>\log C_D+2$.
    For all sufficiently large $n$, the right-hand side of the above inequality is then at most $e^{-s}$ throughout this range, and
    \begin{equation*}
        \sum_{a\log m<s\le m^\gamma} C_D^s\left(\frac{s}{m}\right)^{t_s} \le \sum_{s>a\log m}e^{-s}=o(1).
    \end{equation*}
    Summing \eqref{eq:deterministic-configuration-local-bound} over $s$ thus proves \eqref{eq:deterministic-configuration-local-event}.
\end{proof}

We next bound the number of short cycles by a similar pairing argument.
A cycle of length $\ell$ determines a collection of $\ell$ disjoint pairs of half-edges in the underlying perfect matching, so summing the probabilities of these pair collections over all possible cycles gives an upper bound on the expected number of $\ell$-cycles.
Set
\begin{equation*}
    K_D:=2(D-1),\qquad h_n:=\left\lfloor\frac{\log m}{2\log K_D}\right\rfloor.
\end{equation*}

\begin{lemma}
    \label{lem:deterministic-configuration-short-cycles}
    Let $X_{n,\ell}$ denote the number of cycles of length $\ell$ with distinct vertices in $G_n^{\mathrm{conf}}$.
    Then
    \begin{equation}
        \mathbb P\left(\sum_{\ell=3}^{h_n}X_{n,\ell} \le m^{3/4}\right)=1-o(1). \label{eq:deterministic-few-short-cycles}
    \end{equation}
\end{lemma}

\begin{proof}
    A cycle of length $\ell$ can be specified by choosing its cyclically ordered vertices and, at each vertex, an ordered pair of distinct half-edges.
    For each such choice, the probability that all $\ell$ prescribed pairs occur in the random perfect matching is
    \begin{equation*}
        \prod_{j=0}^{\ell-1}\frac{1}{Dm-(2j+1)}.
    \end{equation*}
    Consequently, uniformly for $\ell=O(\log m)$,
    \begin{equation*}
        \begin{aligned}
            \mathbb E[X_{n,\ell}] &\le \frac{(m)_\ell}{2\ell}\bigl(D(D-1)\bigr)^\ell \prod_{j=0}^{\ell-1}\frac{1}{Dm-(2j+1)}\\
            &\le \frac{m^\ell}{2\ell}\bigl(D(D-1)\bigr)^\ell \left(\frac{2}{Dm}\right)^\ell =\frac{K_D^\ell}{2\ell}.
        \end{aligned}
    \end{equation*}
    It follows that
    \begin{equation*}
        \mathbb E\left[\sum_{\ell=3}^{h_n}X_{n,\ell}\right] =O(K_D^{h_n})=O(m^{1/2}).
    \end{equation*}
    The claim now follows from Markov's inequality.
\end{proof}

The preceding lemmas show that, with high probability, $G_n^{\mathrm{conf}}$ satisfies the required local-sparsity condition and contains at most $m^{3/4}$ cycles of length at most $h_n$.
It remains to choose a simple realization with both properties and remove its short cycles, which proves \cref{lem:deterministic-planted-graphs}.

\begin{proof}[Proof of Lemma \ref{lem:deterministic-planted-graphs}]
    Since $D$ is even, $Dm$ is even for every $m$, so the configuration model is always well defined.
    For fixed $D$, the classical simplicity estimate \cite{Bollobas1980RegularGraphs} gives
    \begin{equation*}
        \mathbb P\bigl(G_n^{\mathrm{conf}}\text{ is simple}\bigr) =\exp\left(-\frac{D^2-1}{4}\right)+o(1)
    \end{equation*}
    as $m\to\infty$.
    In particular, the simplicity probability is bounded away from zero for all sufficiently large $n$.

    By Lemmas \ref{lem:deterministic-configuration-local-sparsity} and \ref{lem:deterministic-configuration-short-cycles}, the required local-sparsity property and short-cycle bound hold simultaneously with probability $1-o(1)$.
    Hence, their intersection with the simplicity event has positive probability for all sufficiently large $n$.
    Choose a simple $D$-regular realization $G_n^{\mathrm{reg}}$ in this intersection.
    For each cycle $C$ in $G_n^{\mathrm{reg}}$ of length at most $h_n$, choose an edge $e_C\in E(C)$, let $\mathcal R_n$ be the set of all chosen edges, and define $\Gamma_n:=G_n^{\mathrm{reg}}\setminus\mathcal R_n$.
    The short-cycle bound \eqref{eq:deterministic-few-short-cycles} yields
    \begin{equation*}
        |\mathcal R_n|\le m^{3/4}.
    \end{equation*}
    Every cycle in $G_n^{\mathrm{reg}}$ of length at most $h_n$ contains an edge in $\mathcal R_n$, so none of these cycles survives in $\Gamma_n$.
    Therefore,
    \begin{equation*}
        \operatorname{girth}(\Gamma_n)>h_n =\left(\frac{r}{2\log K_D}+o(1)\right)\log n.
    \end{equation*}
    Setting $\kappa:=r/(4\log K_D)$ now yields \eqref{eq:deterministic-girth} for all sufficiently large $n$.
    Since deleting edges cannot increase degrees or induced edge counts, $\Gamma_n$ satisfies \eqref{eq:deterministic-max-degree} and \eqref{eq:deterministic-local-sparsity}.
    Moreover, since $G_n^{\mathrm{reg}}$ is $D$-regular and $|\mathcal R_n|\le m^{3/4}$,
    \begin{equation*}
        \frac{D}{2}-m^{-1/4} \le \frac{\mathsf e(\Gamma_n)}{m} \le \frac{D}{2}.
    \end{equation*}
    This proves \eqref{eq:deterministic-edge-density}.
    For the finitely many remaining values of $n$, take $\Gamma_n$ to be an edgeless graph, with the usual convention that a forest has infinite girth.
    The first three conditions then remain valid, while modifying finitely many graphs does not affect the asymptotic relation \eqref{eq:deterministic-edge-density}.
    \end{proof}

\section*{Acknowledgments and Use of AI Tools}
\addcontentsline{toc}{section}{Acknowledgments and Use of AI Tools}
The authors used OpenAI's GPT-5.6 to assist with drafting and revising
portions of the exposition and with generating TikZ code for several
figures. GPT-5.6 also helped identify the correspondence between the Fourier
weights and isomorphism triples (\cref{section:overview} and
\cref{lem:isomorphism-triples}).
In the authors' view, this perspective streamlines the relevant arguments
and clarifies their combinatorial structure, although a similar
estimate can also be obtained through more involved counting arguments. The authors also used Anthropic's Claude
to polish the exposition and the figures. All AI-assisted material was carefully
reviewed, edited, and independently verified by the authors, who take
full responsibility for the contents of the paper.

The authors used OpenAI Codex to formalize in Lean~4 the statements and proofs
of \cref{thm:main,thm:subcritical,thm:deterministic-gap} and of the numbered
results used in their proofs; the development contains no \texttt{sorry} and no
project-specific axioms. The code is available at
\url{https://github.com/fifalsp/lean/tree/6be374ca79dc52351f9cab88ae02692654241c37/2608-low-degree-indistinguishability}.

    \bibliographystyle{alpha}
    \bibliography{ref}
\end{document}

%% file: figures/tree_decomposition_original_graph.tex
\definecolor{tdink}{HTML}{25282C}
\definecolor{tdneutral}{HTML}{66717C}

\begin{tikzpicture}[
    x=0.95cm,
    y=0.95cm,
    line cap=round,
    line join=round,
    graph edge/.style={draw=tdneutral,line width=0.75pt},
    vertex/.style={
        circle,
        draw=tdink,
        fill=white,
        line width=0.55pt,
        minimum size=4.6mm,
        inner sep=0pt,
        font=\scriptsize
    }
]
    \coordinate (p1)  at ( 0.00, 3.80);
    \coordinate (p2)  at (-1.15, 2.80);
    \coordinate (p3)  at ( 1.15, 2.80);
    \coordinate (p4)  at ( 0.00, 2.05);
    \coordinate (p5)  at (-1.00, 0.85);
    \coordinate (p6)  at ( 1.00, 0.85);
    \coordinate (p7)  at ( 0.00,-0.05);
    \coordinate (p8)  at (-2.15,-0.15);
    \coordinate (p9)  at (-3.00,-1.10);
    \coordinate (p10) at (-1.30,-1.10);
    \coordinate (p11) at (-2.15,-2.00);

    \foreach \u/\v in {1/2,1/3,1/4,2/3,2/4,3/4,4/5,4/6,5/6,5/7,6/7,5/8,8/9,8/10,8/11,9/11,10/11}{
        \draw[graph edge] (p\u)--(p\v);
    }
    \foreach \i in {1,...,11}{
        \node[vertex] at (p\i) {\i};
    }
\end{tikzpicture}

%% file: figures/tree_decomposition_stage_one.tex
\definecolor{tdink}{HTML}{25282C}
\definecolor{tdneutral}{HTML}{66717C}
\definecolor{tdlightneutral}{HTML}{F3F5F7}

\begin{tikzpicture}[
    line cap=round,
    line join=round,
    tree edge/.style={draw=tdink,line width=0.7pt},
    coarse bag/.style={
        rounded corners=2pt,
        draw=tdneutral,
        fill=tdlightneutral,
        line width=0.7pt,
        minimum width=34mm,
        minimum height=13mm,
        inner sep=2.4pt,
        align=center,
        font=\scriptsize
    }
]
    \node[coarse bag] (TA) at (0,3.55)
        {$A=\{1,2,3,4\}$\\[-1pt]
         owns $(1,2),(1,3),(1,4)$\\[-1pt]
         $(2,3),(2,4),(3,4)$};
    \node[coarse bag] (TB) at (0,1.75)
        {$B=\{4,5,6\}$\\[-1pt]
         owns $(4,5),(4,6),(5,6)$};
    \node[coarse bag] (TC) at (-2.65,0.05)
        {$C=\{5,6,7\}$\\[-1pt]
         owns $(5,7),(6,7)$};
    \node[coarse bag] (TD) at (2.40,0.05)
        {$D=\{5,8,9,11\}$\\[-1pt]
         owns $(5,8),(8,9)$\\[-1pt]
         $(8,11),(9,11)$};
    \node[coarse bag] (TE) at (2.40,-1.75)
        {$E=\{8,10,11\}$\\[-1pt]
         owns $(8,10),(10,11)$};

    \begin{scope}[on background layer]
        \draw[tree edge] (TB)--(TA);
        \draw[tree edge] (TB)--(TC);
        \draw[tree edge] (TB)--(TD)--(TE);
    \end{scope}

\end{tikzpicture}

%% file: figures/tree_decomposition_stage_two.tex
\definecolor{tdink}{HTML}{25282C}
\definecolor{tdneutral}{HTML}{66717C}
\definecolor{tdblue}{HTML}{2878C8}
\definecolor{tdorange}{HTML}{E07A1F}
\definecolor{tdgreen}{HTML}{2D9348}

\begin{tikzpicture}[
    x=1.45cm,
    y=1cm,
    line cap=round,
    line join=round,
    owner/.style={
        rounded corners=1.3pt,
        line width=0.7pt,
        minimum width=8mm,
        minimum height=4.6mm,
        inner sep=0.8pt,
        font=\scriptsize
    },
    owner blue/.style={owner,draw=tdblue,fill=tdblue!8},
    owner orange/.style={owner,draw=tdorange,fill=tdorange!9},
    owner green/.style={owner,draw=tdgreen,fill=tdgreen!8},
    internal/.style={circle,line width=0.6pt,minimum size=2.1mm,inner sep=0pt},
    internal blue/.style={internal,draw=tdblue,fill=tdblue!35},
    internal orange/.style={internal,draw=tdorange,fill=tdorange!38},
    internal green/.style={internal,draw=tdgreen,fill=tdgreen!35},
    port/.style={
        rectangle,
        line width=0.65pt,
        minimum width=5.0mm,
        minimum height=4.0mm,
        inner sep=0.5pt,
        font=\scriptsize
    },
    port blue/.style={port,draw=tdblue,fill=white},
    port orange/.style={port,draw=tdorange,fill=white},
    port green/.style={port,draw=tdgreen,fill=white},
    blue tree edge/.style={draw=tdblue,line width=0.72pt},
    orange tree edge/.style={draw=tdorange,line width=0.72pt},
    green tree edge/.style={draw=tdgreen,line width=0.72pt},
    boundary edge/.style={draw=tdink,densely dotted,line width=1.0pt},
    same bag/.style={
        rounded corners=3pt,
        draw=tdneutral!80,
        dashed,
        fill=white,
        inner sep=3.5pt,
        line width=0.5pt
    },
    bag label/.style={font=\scriptsize,text=tdink}
]

    \node[internal orange] (a1) at (-1.20,3.95) {};
    \node[internal blue] (a2) at (-0.60,3.95) {};
    \node[internal blue] (a3) at ( 0.00,3.95) {};
    \node[internal blue] (a4) at ( 0.60,3.95) {};
    \node[internal blue] (a5) at ( 1.20,3.95) {};
    \node[owner orange] (Ao12) at (-1.55,4.55) {$(1,2)$};
    \node[port orange]  (ApB)  at (-1.55,3.35) {$p_B$};
    \node[owner blue] (Ao13) at (-0.60,3.35) {$(1,3)$};
    \node[owner blue] (Ao14) at ( 0.00,3.35) {$(1,4)$};
    \node[owner blue] (Ao23) at ( 0.60,3.35) {$(2,3)$};
    \node[owner blue] (Ao24) at ( 1.55,3.35) {$(2,4)$};
    \node[owner blue] (Ao34) at ( 1.55,4.55) {$(3,4)$};
    \node[font=\scriptsize,anchor=east] at (-2.20,4.55) {root};
    \draw[-{Stealth[length=1.8mm,width=1.2mm]},draw=tdink,line width=0.65pt]
        (-2.12,4.55)--($(Ao12.west)+(-0.05,0)$);

    \node[internal orange] (b1) at (-1.00,1.35) {};
    \node[internal orange] (b2) at (-0.33,1.35) {};
    \node[internal orange] (b3) at ( 0.33,1.35) {};
    \node[internal orange] (b4) at ( 1.00,1.35) {};
    \node[port orange]  (BpA)  at (-1.35,1.95) {$p_A$};
    \node[port orange]  (BpC)  at (-1.65,1.35) {$p_C$};
    \node[owner orange] (Bo45) at (-0.33,0.75) {$(4,5)$};
    \node[owner orange] (Bo46) at ( 0.33,0.75) {$(4,6)$};
    \node[port orange]  (BpD)  at ( 1.65,1.35) {$p_D$};
    \node[owner orange] (Bo56) at ( 1.00,0.75) {$(5,6)$};

    \node[internal orange] (c0) at (-2.85,1.35) {};
    \node[owner orange] (Co57) at (-3.35,1.95) {$(5,7)$};
    \node[owner orange] (Co67) at (-3.35,0.75) {$(6,7)$};
    \node[port orange]  (CpB)  at (-2.30,1.35) {$p_B$};

    \node[internal green] (d1) at (-1.00,-1.25) {};
    \node[internal green] (d2) at (-0.33,-1.25) {};
    \node[internal green] (d3) at ( 0.33,-1.25) {};
    \node[internal orange] (d4) at ( 1.00,-1.25) {};
    \node[owner green] (Do58)  at (-1.35,-0.65) {$(5,8)$};
    \node[port green]  (DpE)   at (-1.65,-1.25) {$p_E$};
    \node[owner green] (Do89)  at (-0.33,-1.85) {$(8,9)$};
    \node[owner green] (Do811) at ( 0.33,-1.85) {$(8,11)$};
    \node[port orange]  (DpB)   at ( 1.65,-1.25) {$p_B$};
    \node[owner orange] (Do911) at ( 1.00,-1.85) {$(9,11)$};

    \node[internal green] (e0) at (-2.85,-1.25) {};
    \node[owner green] (Eo810)  at (-3.35,-0.65) {$(8,10)$};
    \node[owner green] (Eo1011) at (-3.35,-1.85) {$(10,11)$};
    \node[port green]  (EpD)    at (-2.30,-1.25) {$p_D$};

    \begin{scope}[on background layer]
        \node[same bag,fit=(a1)(a5)(Ao12)(ApB)(Ao13)(Ao14)(Ao23)(Ao24)(Ao34)]
            (boxA) {};
        \node[same bag,fit=(b1)(b4)(BpA)(BpC)(Bo45)(Bo46)(BpD)(Bo56)]
            (boxB) {};
        \node[same bag,fit=(c0)(Co57)(Co67)(CpB)] (boxC) {};
        \node[same bag,fit=(d1)(d4)(Do58)(DpE)(Do89)(Do811)(DpB)(Do911)]
            (boxD) {};
        \node[same bag,fit=(e0)(EpD)(Eo810)(Eo1011)] (boxE) {};

        \draw[boundary edge] (a1)--(a2);
        \draw[blue tree edge] (a2)--(a3)--(a4)--(a5);
        \draw[orange tree edge] (a1)--(Ao12) (a1)--(ApB);
        \draw[blue tree edge] (a2)--(Ao13) (a3)--(Ao14) (a4)--(Ao23)
            (a5)--(Ao24) (a5)--(Ao34);

        \draw[orange tree edge] (b1)--(b2)--(b3)--(b4);
        \draw[orange tree edge] (b1)--(BpA) (b1)--(BpC)
            (b2)--(Bo45) (b3)--(Bo46) (b4)--(BpD) (b4)--(Bo56);
        \draw[orange tree edge] (c0)--(Co57) (c0)--(Co67) (c0)--(CpB);

        \draw[green tree edge] (d1)--(d2)--(d3);
        \draw[boundary edge] (d3)--(d4);
        \draw[green tree edge] (d1)--(Do58) (d1)--(DpE)
            (d2)--(Do89) (d3)--(Do811);
        \draw[orange tree edge] (d4)--(DpB) (d4)--(Do911);
        \draw[green tree edge] (e0)--(EpD) (e0)--(Eo810) (e0)--(Eo1011);

        \draw[orange tree edge] (ApB)--(BpA);
        \draw[orange tree edge] (BpC)--(CpB);
        \draw[orange tree edge] (BpD)--(DpB);
        \draw[green tree edge]  (DpE)--(EpD);
    \end{scope}

    \foreach \x/\bag in {A/{1,2,3,4},B/{4,5,6},C/{5,6,7},D/{5,8,9,11},E/{8,10,11}}{
        \node[bag label,anchor=south] at ($(box\x.north)+(0,0.04)$)
            {$R_\x\colon B_\x=\{\bag\}$};
    }
\end{tikzpicture}

%% file: figures/tree_decomposition_stage_three.tex
\definecolor{tdink}{HTML}{25282C}
\definecolor{tdblue}{HTML}{2878C8}
\definecolor{tdorange}{HTML}{E07A1F}
\definecolor{tdgreen}{HTML}{2D9348}

\begin{tikzpicture}[
    line cap=round,
    line join=round,
    vertex/.style={
        circle,
        draw=tdink,
        fill=white,
        line width=0.55pt,
        minimum size=4.6mm,
        inner sep=0pt,
        font=\scriptsize
    },
    note/.style={font=\scriptsize,text=tdink,align=center}
]
    \coordinate (p1)  at ( 0.00, 3.80);
    \coordinate (p2)  at (-1.15, 2.80);
    \coordinate (p3)  at ( 1.15, 2.80);
    \coordinate (p4)  at ( 0.00, 2.05);
    \coordinate (p5)  at (-1.00, 0.85);
    \coordinate (p6)  at ( 1.00, 0.85);
    \coordinate (p7)  at ( 0.00,-0.05);
    \coordinate (p8)  at (-2.15,-0.15);
    \coordinate (p9)  at (-3.00,-1.10);
    \coordinate (p10) at (-1.30,-1.10);
    \coordinate (p11) at (-2.15,-2.00);

    \foreach \u/\v in {1/3,1/4,2/3,2/4,3/4}{
        \draw[draw=tdblue,line width=1.05pt] (p\u)--(p\v);
    }

    \foreach \u/\v in {1/2,4/5,4/6,5/6,5/7,6/7,9/11}{
        \draw[draw=tdorange,line width=1.05pt] (p\u)--(p\v);
    }

    \foreach \u/\v in {5/8,8/9,8/10,8/11,10/11}{
        \draw[draw=tdgreen,line width=1.05pt] (p\u)--(p\v);
    }

    \foreach \i in {1,...,11}{
        \node[vertex] at (p\i) {\i};
    }

    \foreach \col/\name/\num [count=\r from 0] in
        {tdblue/T_1/5,tdorange/T_2/7,tdgreen/T_3/5}{
        \draw[draw=\col,line width=1.2pt]
            (-0.75,-0.95-0.45*\r)--(-0.25,-0.95-0.45*\r);
        \node[font=\scriptsize,anchor=west] at (-0.15,-0.95-0.45*\r)
            {$\name$ (\num\ edges)};
    }

    \node[note] at (-0.90,-2.75)
        {$E(F)=E(T_1)\sqcup E(T_2)\sqcup E(T_3)$};
\end{tikzpicture}